\documentclass[a4paper]{article}

\usepackage{fullpage}
\usepackage{authblk,placeins}
\usepackage{microtype}
\usepackage{amsthm,thmtools,thm-restate,soul,wrapfig}
\usepackage{enumerate,amsfonts,graphicx,url,stfloats,amssymb,amsmath,color,verbatim, tikz,multirow}

\usetikzlibrary{positioning,patterns,snakes}
\usepackage{environ}

\usepackage{paralist,cuted}
\usepackage[utf8]{inputenc}
\usepackage[T1]{fontenc}

\declaretheorem{theorem}
\declaretheorem[sibling=theorem]{definition}
\declaretheorem[sibling=theorem]{lemma}
\declaretheorem[sibling=theorem]{example}
\declaretheorem[sibling=theorem]{conjecture}

\declaretheorem[sibling=theorem]{remark}

\usepackage{graphicx,subcaption}
\usepackage{caption,times,fullpage}

\usetikzlibrary{calc}

\newcommand{\masterA}{{\cal A}_{mas}}
\newcommand{\nestedA}{\mathbb{A}}
\newcommand{\slaveA}{{\mathfrak{B}}}
\newcommand{\nonnestedA}{{\cal A}}
\newcommand{\selectA}{{\cal S}}

\newcommand{\silent}[1]{\mathsf{sil}({#1})}

\newcommand{\cost}{{C}}

\newcommand{\masterRun}{\Pi}
\newcommand{\slaveRun}{\pi}
\newcommand{\lang}{{\cal L}}

\newcommand{\valueL}[1]{{\cal L}_{{#1}}}

\newcommand{\abs}{\mathop{\mathsf{Abs}}}

\newcommand{\lpair}[2]{\langle {#1}, {#2} \rangle}
\newcommand{\triple}[3]{\langle {#1}, {#2}, {#3} \rangle}
\newcommand{\quadruple}[4]{\langle {#1}, {#2}, {#3}, {#4} \rangle}

\newcommand{\PTIME}{\textsc{PTime}{}}
\newcommand{\EXPSPACEshort}{\textsc{ExpSp.}{}}
\newcommand{\PSPACE}{\textsc{PSpace}{}}
\newcommand{\EXPSPACE}{\textsc{ExpSpace}{}}
\newcommand{\NLOGSPACE}{\textsc{NLogSpace}{}}
\newcommand{\PSPACEshort}{\textsc{PSp.}{}}
\newcommand{\N}{\mathbb{N}}

\newcommand{\Q}{\mathbb{Q}}

\newcommand{\hide}[1]{}

\definecolor{darkgreen}{RGB}{00,140,00}
\definecolor{darkred}{RGB}{140,0,00}
\definecolor{darkblue}{RGB}{0,0,140}

\newcommand{\undecidable}{\textcolor{darkred}{Undec.}}

\newcommand{\conjectureOne}{\textcolor{darkgreen}{Open \eqref{conj1}}}

\newcommand{\buchi}{B\"{u}chi}

\newcommand{\fsum}{\textsc{Sum}}
\newcommand{\fBsum}[1]{\textsc{Sum}^{#1}}
\newcommand{\fmax}{\textsc{Max}}
\newcommand{\fmin}{\textsc{Min}}

\newcommand{\flimavg}{\textsc{LimAvg}}
\newcommand{\fliminf}{\textsc{LimInf}}
\newcommand{\flimsup}{\textsc{LimSup}}

\newcommand{\fsup}{\textsc{Sup}}
\newcommand{\finf}{\textsc{Inf}}

\newcommand{\aut}{{\cal A}}
\newcommand{\Run}{\mathsf{Run}}
\newcommand{\Acc}{\mathsf{Acc}}
\newcommand{\const}{\lambda}
\newcommand{\FinVal}{\mathsf{FinVal}}
\newcommand{\InfVal}{\mathsf{InfVal}}

\newcommand{\tabref}[1]{\eqref{#1}}

\newcommand{\tikzSlaveA}{
\node[circle,draw] (P1) at (0,0) {$q_a$};
\node[circle,draw] (P2) at (1.8,0){$q_F$};
\node[circle,draw, minimum size=0.5cm] (P0) at (1.8,0){};

\draw[->] (P1) to[loop left] node[above,inner sep=5pt] {$(a,1)$} (P1);

\draw[->] (P1) to node[above] {$(b,0)$} (P2);
}

\newcommand{\tikzMasterOne}{
\node[circle, draw] (Q0) at (0,2.5) {$q_0$};

\draw[->] (Q0) to[loop above] node[above] {$(a,1)$} (Q0);
\draw[->] (Q0) to[loop below] node[below] {$(b,2)$} (Q0);
}

\newcommand{\tikzMasterTwo}{
\node[circle,draw] (Q1) at (0,2.5) {$q_a$};
\node[circle,draw] (Q2) at (1.8,2.5){$q_0$};

\draw[->] (Q1) to[loop above] node {$(a,3)$} (Q1);
\draw[->] (Q2) to[loop above] node {$(b,3)$} (Q2);

\draw[->] (Q1) to[bend left] node[above] {$(b,2)$} (Q2);
\draw[->] (Q2) to[bend left] node[below] {$(a,1)$} (Q1);
}

\newcommand{\tikzOfSlaveRun}{
\foreach \i/\l in {1/a,2/a,3/a,4/b,5/a,6/a,7/a,8/b}
{
   \node[rectangle,draw,minimum height=0.5cm] (a\i) at (\i*0.4,0) {$\l$};
}
\foreach \i/\j/\l in {1/2/3,2/3/2,3/4/1,4/5/1,5/6/3}
{
   \draw[->] ($(a\i.north) + (0.03,0)$) to[out=90,in=90] node[above] (b\i) {\l} ($(a\j.north)+ (-0.03,0)$) ;
}
\foreach \i/\j/\l in {1/6/0,2/6/1,3/6/2,4/6/2,
                      2/5/0,3/5/1,4/5/1,
                      3/4/0,4/4/0,
                      4/3/0,5/3/0,
                      5/2/0,6/2/1,7/2/2,8/2/2}
{
   \node[rectangle,draw] (c\i\j) at (\i*0.4,\j*0.55+0.2) {$\l$};
}
\foreach \i/\j/\k in {1/6,2/5,3/4,4/3,5/2}
{
   \draw[->] (c\i\j)  to[out=270,in=110] (b\i) ;
}
\foreach \i/\j/\k in {5/6/4,6/5/3,7/4/2,9/3/2,13/2/4}
{
 \node[rectangle,draw, minimum width=\k*0.4cm+0.1cm,minimum height=0.55cm,thick] at (\i*0.2,\j*0.55+0.2) {};
}
\foreach \i/\j/\l/\w/\c in {6/6/\slaveA_1/a^3b/3,6/5/\slaveA_1/aab/2,6/4/\slaveA_1/ab/1,7/3/\slaveA_2/ba/1}
{
   \node[rectangle] (c\i\j) at (\i*0.43+0.35,\j*0.55+0.2) {$\lang_{\l}(\w)=\c$};
}
}

\newcommand{\confNum}[1]{\mathsf{conf}({#1})}

\begin{document}
\title{Nested Weighted Automata
\thanks{This research was funded in part by the European Research Council (ERC) under grant 
agreement 267989 (QUAREM), by the Austrian Science Fund (FWF) projects S11402-N23 (RiSE),
Z211-N23 (Wittgenstein Award), FWF Grant No P23499- N23, FWF NFN Grant No S11407-N23 (RiSE), 
ERC Start grant (279307: Graph Games), and Microsoft faculty fellows award.}
}

\author{Krishnendu Chatterjee}
\author{Thomas A. Henzinger}
\author{Jan Otop}
\affil{IST Austria}

\maketitle

\begin{abstract}
Recently there has been a significant effort to handle quantitative properties 
in formal verification and synthesis.
While weighted automata over finite and infinite words provide a natural and
flexible framework to express quantitative properties, perhaps surprisingly,
some basic system properties such as average response time cannot be
expressed using weighted automata, nor in any other know decidable formalism.
In this work, we introduce nested weighted automata as a natural extension 
of weighted automata which makes it possible to express important 
quantitative properties such as average response time.
In nested weighted automata, a master automaton spins off and collects 
results from weighted slave automata, each of which computes a quantity along a finite 
portion of an infinite word.
Nested weighted automata can be viewed as the quantitative analogue of 
monitor automata, which are used in run-time verification.
We establish an almost complete decidability picture for the basic decision
problems about nested weighted automata, and illustrate their applicability 
in several domains. In particular, nested weighted automata can be used to decide 
average response time properties.   
\end{abstract}

\section{Introduction}
\label{sec:intro}

Traditionally, formal verification has focused on Boolean properties 
of systems, such as ``every request is eventually granted.'' 
Automata-theoretic formalisms as well as temporal logics have been 
studied as specification languages for such Boolean properties of 
reactive systems.  In recent years there has been a growing trend to 
extend specifications with quantitative aspects for expressing 
properties such as ``the long-run average success rate of an operation 
is at least one half'' or ``the long-run average (or the maximal, or 
the accumulated) resource consumption is below a threshold.''
Quantitative aspects of specifications are essential for 
resource-constrained systems, such as embedded systems, and for 
performance evaluation.  For example, quantitative specifications such 
as accumulated sum can express properties like number of $a$ events
between $b$ events, or memory consumption, whereas long-run average
can express properties related to reliability requirements such as
average success rate.

For Boolean properties regular languages provide a robust 
specification framework. Finding the analogue of regular 
languages for quantitative specifications is an active 
research area~\cite{DBLP:conf/lics/AlurDDRY13,Chatterjee08quantitativelanguages,Droste:2009:HWA:1667106}.
Some of the key features of such a specification framework are 
(1)~expressiveness, i.e., whether the formalism can express properties 
of interest; (2)~ease of specification, i.e., whether the properties 
can be stated naturally; (3)~computability, i.e., whether the basic
decision problems can be solved ---ideally with elementary
complexity--- for interesting fragments of the formalism; and
(4)~robustness, i.e., whether the formalism is robust against small
changes in its definition.

While automata are an expressive, natural, elementarily decidable, and
robust framework for expressing Boolean properties, weighted automata
provide a natural and flexible framework for expressing
quantitative\footnote{We use the term ``quantitative'' in a
non-probabilistic sense, which assigns a quantitative value to each
infinite run of a system, representing long-run average or maximal
response time, or power consumption, or the like, rather than taking a
probabilistic average over different runs.}
properties~\cite{Chatterjee08quantitativelanguages}.  Weighted
automata are an extension of finite automata in which every transition
is labeled by a rational weight.  Thus, each run produces a sequence
of weights, and a value function aggregates the sequence into a single
value.  For non-deterministic weighted automata, the value of a word
$w$ is the infimum value of all runs over~$w$.  Initially, weighted
automata were studied over finite words with weights from a semiring,
and ring multiplication as value function~\cite{Droste:2009:HWA:1667106}.
They have been extended to infinite words with limit averaging or
supremum as value function~\cite{Chatterjee:2009:AWA:1789494.1789497,DBLP:journals/corr/abs-1007-4018,Chatterjee08quantitativelanguages}. 
While weighted automata over semirings can express several
quantitative properties~\cite{DBLP:journals/jalc/Mohri02}, they cannot
express the following basic quantitative properties.

\begin{example}\label{ex:intro}
Consider infinite words over $\{r,g,i\}$, where $r$ represents 
requests, $g$ represents grants, and $i$ represents idle. A first 
basic property is the long-run average frequency of $r$'s, which
corresponds to the average workload of a system.  A second interesting
property is the average number of $i$'s between a request and the
corresponding grant, which represents the long-run average response time of the
system.
\end{example}

While weighted automata with limit-average as value function can 
express the average workload property (which weighted automata over semirings cannot express), 
perhaps surprisingly, they are not capable of expressing the long-run average response time.
To see this, notice that the value of a weighted automaton with 
limit-average value function is bounded by the maximal weight that 
occurs in the automaton, whereas the long-run average response time can be 
unbounded. However, long-run average response time can be 
computed if the sum value function can be applied between requests and 
subsequent grants, and the values of the sum function can be 
aggregated using limit-average function. Such a mechanism can be 
expressed naturally by an extension of weighted automata, 
called \emph{nested weighted automata}, which we introduce in this 
paper.

A nested weighted automaton consists of a master automaton and a set 
of slave automata. The master automaton runs over an infinite word, 
and at each transition of the infinite run, it may invoke a slave 
automaton that runs over a finite subword of the infinite word, 
starting from the position where the master automaton invokes the 
slave automaton. Each slave automaton terminates after a finite 
number of steps and returns a value to the master automaton. To 
compute its return value, each slave automaton is equipped with a
value function for finite words, and the master automaton aggregates
all return values using a value function for infinite words.  While in
the case of Boolean finite automata, nested automata are no more
expressive than their non-nested counterpart, we show that  
the class of nested weighted automata is strictly more expressive 
than non-nested weighted automata.  
For example, with nested weighted automata, the long-run average response 
time of a word can be computed, as in the following example.

\begin{example}
In Example~\ref{ex:intro} there is only a single type of request and grant, 
but in general there can be multiple types of requests and grants, and the 
intervals between requests and corresponding grants for different
requests may overlap.  Using a nested weighted automaton, the average
response time can be specified across all requests.  We illustrate this for
two types of requests and corresponding grants.  The input alphabet is
$\{r_1,g_1,r_2,g_2,i\}$.  At every request $r_1$ (resp.~$r_2$) the
master automaton spins off a slave automaton $\slaveA_1$
(resp.~$\slaveA_2$) with a sum value function, which counts the number
of idle events to the next corresponding grant $g_1$ (resp.~$g_2$).
Observe that many slave automata may run concurrently.  Indeed, for
the word $(r_1^n r_2^n g_1 g_2)^{\omega}$, at all positions with the
letter $g_1$ there are $2 \cdot n$ slave automata that run
concurrently.  The master automaton with limit-average value
function then averages the response times returned by the slave
automata.
\end{example}

Our contributions are three-fold.  First, we introduce nested weighted
automata over infinite words (Section~\ref{s:nested-automata}), which
is a new formalism for expressing important quantitative properties,
such as long-run average response time, which cannot be specified by
non-nested weighted automata.

Second, we study the decidability and complexity of emptiness, universality, 
and inclusion for nested weighted automata. We present an almost complete 
decidability picture for several natural and well-studied value 
functions. 
\begin{compactitem}
\item
On the positive side, we show that if the value functions 
of the slave automata are max, min, or bounded sum, then the decision 
problems for nested weighted automata can be reduced to the 
corresponding problems for non-nested weighted automata. Moreover, we 
show that if the value function of the master automaton is limit 
average and the value function of the slave automata is non-negative 
sum (i.e., sum over non-negative weights), which includes the long-run 
average response time property, then the emptiness question is 
decidable in exponential space. 
Along with the decidability results, we also establish parametric complexity
results, that show that when the total size of the slave automata is 
bounded by constant (which is the case for average response property), 
then for all decidability results the complexity matches that of Boolean 
non-nested automata (see Remark~\ref{rem:parametric}).
The decidability proof is obtained by establishing certain regularity 
properties of optimal runs, which can be used to reduce the problem to the 
emptiness question for non-nested weighted automata with limit-average 
value function. 
\item 
On the negative side, we show that even for deterministic nested weighted 
automata with sup value function for the master automaton and sum 
value function for the slave automata, the emptiness question is 
undecidable. This result is in sharp contrast to non-nested weighted 
automata, where the emptiness and universality questions are always 
decidable for deterministic automata, and the emptiness question is 
decidable also for non-deterministic sup and sum automata. 
\end{compactitem}
Our results are summarized in Table~\ref{tab1} and Table~\ref{tab2}
(in Table~\ref{tab3} and Table~\ref{tab4} for parametric complexity results) 
in Section~\ref{subsec:summary}.

Third, nested weighted automata provide a convenient formalism to
express quantitative properties.  In the Boolean case, \emph{monitor
automata} offer a natural compositional way of specifying complex
temporal properties~\cite{DBLP:conf/spin/PnueliZ08}, and nested
weighted automata can be seen as a quantitative extension of monitor
automata.  Each monitor automaton tracks a subproperty (which
corresponds to a slave automaton in our formalism), and the results of
the monitor automata are combined by another monitor automaton (which
corresponds to the master automaton in our formalism).
\begin{itemize}
\item 
A key advantage of the monitor-automaton approach is that it allows
complex specifications to be decomposed into subproperties, which
eases the task of specification.  Our nested weighted automata enjoy
the same advantage: e.g., for long-run average response time, each
slave automaton computes the response time (as a sum automaton) of an
event, and the master automaton averages the response times.
Formally, we show that deterministic nested weighted automata can be
exponentially more succinct than non-deterministic weighted automata
even when they express the same property (Theorem~\ref{thm:succinct}).  
Moreover, monitor automata are used heavily in run-time
verification~\cite{DBLP:conf/tacas/HavelundR02}.  Hence our framework
can also be seen as a first step towards quantitative run-time
verification, where the slave automata return values of subproperties,
and the master automaton (assuming a commutative value function)
computes on-the-fly an approximation of the value.
\item 
More importantly, for Boolean properties monitor automata simply 
provide a more convenient framework for specification, as they are 
equally expressive as standard automata, whereas we show that nested 
weighted automata are strictly more expressive than non-nested 
weighted automata (e.g., long-run average response time, which cannot 
be expressed using non-nested weighted automata, can be expressed using
nested weighted automata.
\end{itemize}
In other words, we provide a natural combination of weighted automata
(for quantitative properties) and nesting of automata (for ease of 
expressiveness), and as a result obtain a more expressive, 
elementarily decidable, and convenient quantitative specification 
framework.

Finally, we illustrate the applicability of nested weighted automata 
in several domains. (1)~We show that the \emph{model-measuring} 
problem of~\cite{myconcur} can be expressed in the nested weighted 
automaton framework (Section~\ref{s:applications-model-measuring}).
The model-measuring problem asks, given a model and a specification, 
how robustly the model satisfies the specification, i.e., how much the 
model can be perturbed without violating the specification. (2)~As 
dual of the model-measuring problem, we introduce the 
\emph{model-repair} problem  and show that it, as well, can be solved using 
nested weighted automata (Section~\ref{s:applications-inner-model-measuring}).
The model-repair problem asks, given a specification and a model that does 
not satisfy the specification, for the minimal restriction of the model 
that satisfies the specification.
We show that we need nested weighted automata in order to express interesting 
measures on models for the model-measuring and model-repair problems. 

In summary, we introduce nested weighted automata, which offer an 
expressive and convenient quantitative specification framework, and 
establish that the basic verification problems are decidable for 
several interesting fragments (which include the long-run average response time 
property).
While there exist many frameworks to express quantitative properties 
(that we discuss in details in Section~\ref{sec:related}), there exists no 
framework (to the best of our knowledge) that can express the average response 
time property and admit algorithms with elementary time complexity for the 
basic decision problems. 
We present a framework (of nested weighted automata) that can express 
such basic system properties and have decidable algorithms with elementary 
complexity.

The paper is a full version of~\cite{thisPaperLICS}.

\section{Preliminaries}
\label{s:prelim}
\noindent{\em Words.} 
Given a finite alphabet $\Sigma$ of letters, a 
finite (resp. infinite) word $w$ is a finite (resp. infinite)
sequence of letters.
For a word $w$ and $i,j \in \N$, we define $w[i]$ as the $i$-th letter of $w$
and $w[i,j]$ as the word $w[i] w[i+1] \ldots w[j]$. 
We allow $j$ to be $\infty$ for infinite words.
For a finite word $w$, we denote by $|w|$ its length; and for an
infinite word the length is $\infty$.

\smallskip\noindent{\em Non-deterministic automata.}
A \emph{(non-deterministic) automaton} $\aut$ is a tuple $(\Sigma, Q, Q_0, \delta, {F})$, 
where $\Sigma$ is the alphabet, $Q$ is a finite set of states, 
$Q_0 \subseteq Q$ is a set of initial states, $\delta \subseteq  Q \times \Sigma \times Q$ 
is a transition relation, and ${F} \subseteq Q$ is a set of \emph{accepting} states.

\smallskip\noindent{\em Runs.} 
Given an automaton $\aut$ and a word $w$, a \emph{run} $\pi= \pi[0] \pi[1] \ldots$ 
is a sequence of states such that $\pi[0] \in Q_0$ and for every 
$i \in \{1, \dots, |w|\}$ we have $(\pi[i-1], w[i] ,\pi[i]) \in \delta$.
Given a word $w$, we denote by $\Run(w)$ the set of all possible runs on $w$.

\smallskip\noindent{\em Boolean acceptance.}
The acceptance of words is defined using the accepting states.
A finite run $\pi$ of length $j+1$ is \emph{accepting} if $\pi[j] \in F$;
and an infinite run $\pi$ is \emph{accepting}, if there exist infinitely many 
$j$ such that $\pi[j] \in F$.
Let $\Acc(w) \subseteq \Run(w)$ denote the set of accepting runs.
A word $w$ is accepted iff $\Acc(w)$ is non-empty.
We denote by $\lang_{\aut}$ the set of words accepted by $\aut$.

\smallskip\noindent{\em Labeled and weighted automata.}
Given a finite alphabet $\Gamma$, a \emph{$\Gamma$-labeled automaton} is an 
automaton whose transitions are labeled by elements from $\Gamma$. 
Formally, a $\Gamma$-labeled automaton $\aut$ is a tuple 
$(\Sigma, Q, Q_0, \delta, F, \cost)$ such that $(\Sigma, Q, Q_0, \delta, F)$ 
is an automaton and $\cost : \delta \mapsto \Gamma$. 
A \emph{weighted} automaton is a $\Gamma$-labeled automaton, where $\Gamma$ is a finite subset of rationals;
and the labels of the transitions are referred to as \emph{weights}.

\smallskip\noindent{\em Semantics of weighted automata.} 
To define the semantics of weighted automata we need to define the value of a 
run (that combines the sequence of weights of a run to a single value) and the 
value across runs (that combines values of different runs to a single value).
To define values of runs, we will consider  \emph{value functions} $f$ that 
assign real numbers to sequences of rationals.
Given a non-empty word $w$, every run $\pi$ of $\aut$ on $w$ defines a sequence of weights 
of successive transitions of $\aut$, i.e., 
$\cost(\pi)=(\cost(\pi[i-1], w[i], \pi[i]))_{1\leq i \leq |w|}$; 
and the value $f(\pi)$ of the run $\pi$ is defined as $f(\cost(\pi))$.
We will denote by $(\cost(\pi))[i]$ the cost of the $i$-th transition,
i.e., $\cost(\pi[i-1], w[i], \pi[i])$.
The value of a non-empty word $w$ assigned by the automaton $\aut$, denoted by  $\valueL{\aut}(w)$,
is the infimum of the set of values of all {\em accepting} runs;
i.e., $\inf_{\pi \in \Acc(w)} f(\pi)$, and we have the usual semantics that infimum of an
empty set is infinite, i.e., the value of a word that has no accepting runs is infinite.
Every run $\pi$ on an empty word has length $1$ and the sequence $\cost(\pi)$ is empty, hence 
we define the value $f(\pi)$ as an external (not a real number) value $\bot$. 
Thus, the value of the empty word is either $\bot$, if the empty word is accepted by $\aut$, or $\infty$ 
otherwise.
To indicate a particular value function $f$ that defines the semantics,
we will call a weighted automaton $\aut$ an $f$-automaton. 

\smallskip\noindent{\em Types of automata.}
A weighted automaton is 
\begin{compactitem}
\item \emph{deterministic} iff $Q_0$ is singleton and the transition relation is a function; and
\item \emph{functional} iff for every word $w$, all accepting runs on $w$ have
 the same value. 
\end{compactitem}

\smallskip\noindent{\em Value functions.}
We will consider the classical functions and their natural variants for
value functions. 
For finite runs we consider the following value functions: for runs of length $n+1$ we have
\begin{compactenum}
\item {\em Max and min:} 
\begin{compactitem}
\item $\fmax(\pi) = \max_{i=1}^n (\cost(\pi))[i]$ and 
\item $\fmin(\pi) = \min_{i=1}^n (\cost(\pi))[i]$.
\end{compactitem}
\item \emph{Sum and variants:} 
\begin{compactitem}
\item the sum function $\fsum(\pi) = \sum_{i=1}^{n} (\cost(\pi))[i]$, 
\item the absolute sum $\fsum^+(\pi) = \sum_{i=1}^{n} \abs((\cost(\pi))[i])$ 
is the sum of the absolute values of the weights ($\abs$ denotes the 
absolute value of a number), and
\item the bounded sum value function returns the sum if all the 
partial absolute sums are below a bound $B$, otherwise it returns the bound 
$B$, i.e., formally, $\fBsum{B}(\pi) = \fsum(\pi)$, 
if for all prefixes $\pi'$ of $\pi$ we have $\abs(\fsum(\pi')) \leq B$, 
otherwise $B$.
\end{compactitem}
\end{compactenum}
We denote the above class of value functions for finite words as 
$\FinVal=\{\fmax,\fmin,\fsum,\fsum^+,\fBsum{B}\}$.

Although, the absolute sum value function $\fsum^+$ can be equivalently expressed by $\fsum$
restricted to the class of weighted automata with non-negative weights,
we consider $\fsum$ and $\fsum^+$ separately, as the resulting automata differ in complexity results. 

For infinite runs we consider:
\begin{enumerate}
\item {\em Supremum and Infimum, and 
Limit supremum and Limit infimum}: 
\begin{compactitem}
\item $\fsup(\pi) = \sup \{ (\cost(\pi))[i] : i > 0 \}$, 
\item $\finf(\pi) = \inf \{ (\cost(\pi))[i] : i > 0 \}$,  
\item $\flimsup(\pi) = \lim\sup \{ (\cost(\pi))[i] : i > 0\}$, and
\item $\fliminf(\pi) = \lim\inf \{  (\cost(\pi))[i] : i > 0 \}$.
\end{compactitem}
\item {\em Limit average:} $\flimavg(\pi) = \limsup\limits_{k \rightarrow \infty} \frac{1}{k} \cdot \sum_{i=1}^{k} (\cost(\pi))[i]$.
\end{enumerate}
We denote the above class of value functions for infinite words as 
$\InfVal=\{\fsup,\finf,\flimsup,\fliminf,\flimavg\}$.

\smallskip\noindent{\em Decision questions.} We consider the standard emptiness and 
universality questions.
Given an $f$-automaton $\aut$ and a threshold $\const$,
the \emph{emptiness} (resp. \emph{universality}) question asks
whether there exists a non-empty word $w$ such that  $\valueL{\aut}(w) \leq \const$
(resp. for all non-empty words $w$ we have $\valueL{\aut}(w) \leq \const$).
We summarize the main results from the literature related to the decision 
questions of weighted automata for the class of value functions defined 
above.

\begin{theorem}
(1)~The emptiness problem is decidable in polynomial time for all value functions we 
consider~\cite{Filar:1996:CMD:248676,DBLP:journals/jalc/Mohri02}.
(2)~The universality problem is undecidable for $\fsum$-automata with 
$\{ -1, 0, 1\}$ weights and $\flimavg$-automata with $\{ 0, 1\}$ weights;
and decidable in polynomial space for all other value functions~\cite{DBLP:conf/atva/AlmagorBK11,DBLP:conf/csl/DegorreDGRT10,DBLP:journals/corr/abs-1006-1492,Leung1991137}.
(3)~The universality problem is decidable for all value functions 
for deterministic and functional automata~\cite{DBLP:conf/concur/FiliotGR12}.
\label{t:oldResults}
\end{theorem}

\section{Nested Weighted Automata}
\label{s:nested-automata}
\newcommand{\idx}{\textsf{idx}}
\newcommand{\stu}{\mathit{stu}}
\newcommand{\pos}{\mathit{pos}}

In this section we introduce nested weighted automata.
We start with an informal description.

\smallskip\noindent{\em Informal description.}
A \emph{nested weighted automaton} consists of a labeled automaton over infinite words, 
called the \emph{master automaton}, a value function $f$,  
and a set of weighted automata over finite words,
called \emph{slave automata}. 
A nested weighted automaton can be viewed as follows: 
given an infinite word, we consider a run of the master automaton on the word,
but the weight of each transition is determined by dynamically running 
slave automata; and then the value of a run is obtained using the 
value function $f$.
That is, the master automaton proceeds on an input word as a usual automaton, 
except that before it takes a transition, it starts a slave automaton 
corresponding to the label of the current transition. 
The slave automaton starts at the current position of the word of the master automaton 
and runs on some finite part of the input word. Once a slave automaton terminates,
it returns its value to the master automaton, which treats the returned
value as the weight of the current transition that is being executed.
Note that for some transitions the slave automaton runs on the empty word and returns $\bot$;
we refer to such transitions as \emph{silent} transitions.
A given run of a nested weighted automaton, which consists of a run of the master automaton and 
runs of slave automata, is accepting if it consists of accepting runs only.
Finally, the value of an accepting run of the master automaton is given by $f$ applied to the 
sequence of values returned by slave automata (i.e., to compute the value function
the silent transitions are omitted).

\smallskip\noindent{\em Nested weighted automata.}
A \emph{nested weighted automaton} is a tuple $\nestedA = \langle \masterA; f;  \slaveA_1, \ldots, \slaveA_k \rangle$, 
where $\masterA$ is a $\{ 1, \ldots, k\}$-labeled automaton over infinite words (where labels are indices of slave automata),
called the \emph{master} automaton, $f \in \InfVal$ is a value function on infinite sequences, and
$\slaveA_1, \ldots, \slaveA_k$ are weighted automata over finite words, 
called \emph{slave} automata.

\smallskip\noindent{\em Semantics: runs and values.}
Let $w$ be an infinite word. 
A \emph{run} of $\nestedA$ on $w$ is an infinite sequence 
$(\masterRun, \slaveRun_1, \slaveRun_2, \ldots)$ such that 
(i)~$\masterRun$ is a run of $\masterA$ on $w$;
(ii)~for every $i>0$ we have $\slaveRun_i$ is a run of the automaton $\slaveA_{\cost(\masterRun[i-1], w[i], \masterRun[i])}$,
referenced by the label $\cost(\masterRun[i-1], w[i], \masterRun[i])$ of the master automaton, on some finite subword $w[i,j]$ of $w$.
The run $(\masterRun, \slaveRun_1, \slaveRun_2, \ldots)$ is accepting if all 
runs $\masterRun, \slaveRun_1,  \slaveRun_2, \ldots$ are accepting (i.e., $\masterRun$ satisfies its acceptance 
condition and each $\slaveRun_1,\slaveRun_2, \ldots$ ends in an accepting state)
and infinitely many runs of slave automata have length greater than $1$ (the master automaton takes infinitely many non-silent transitions).
The value of the run $(\masterRun, \slaveRun_1, \slaveRun_2, \ldots)$ is defined as 
$\silent{f}( v(\pi_1) v(\pi_2) \ldots)$, where $v(\pi_i)$ is the value of the run $\pi_i$ in 
the corresponding slave automaton and 
$\silent{f}$ is the value function that applies $f$ on sequences after removing $\bot$ symbols.
The value of a word $w$ assigned by the automaton $\nestedA$, denoted by  
$\valueL{\nestedA}(w)$, is the infimum of the set of values of all {\em accepting} runs.
We require accepting runs to contain infinitely many non-silent transitions because
$f$ is a value function over infinite sequences, so we need 
the sequence $v(\pi_1) v(\pi_2) \ldots$ with $\bot$ symbols removed to be infinite.

\smallskip\noindent{\em Notation.} Let $f,g$  be value functions.
We say that a nested weighted automaton $\nestedA = \langle \masterA; h;  \slaveA_1, \ldots, \slaveA_k \rangle$ 
is an $(f;g)$-automaton iff $h=f$ and $\slaveA_1, \ldots, \slaveA_k$ are $g$-automata
(weighted automata over finite words with value function $g$).
We illustrate the semantics of nested weighted automata with examples.

\begin{example}[Stuttering]
\label{ex:avg-resp-time-pre}
Consider the nested weighted automaton 
$\nestedA_{\stu}^1 = \langle \masterA^1; \flimavg;  \slaveA_{1}, \slaveA_{2}  \rangle$ 
where each slave automaton is a $\fsum^+$-automaton.
The automaton $\masterA^1$ has a single state and two transitions $(q_0, a, q_0)$ 
labeled by $1$ and $(q_0, b, q_0)$ labeled by $2$. 
The slave automaton $\slaveA_1$ accepts words from $a^*b$ and assigns to a word
$a^k b$ value $k$.
The automaton $\slaveA_2$ accepts words from $b^*a$ and assigns to 
a word $b^k a$ value $k$.

Consider a word $(aaab)^{\omega}$.  A run of $\nestedA_{\stu}^1$ on $(aaab)^{\omega}$ is depicted in 
Figure~\ref{fig:run}. The value of the word is $\frac{7}{4}$. 
Note that $\nestedA_{\stu}^1$ accepts only words with infinite number of $a$'s and $b$'s,
as otherwise, some slave automaton would not terminate.
For word $w = (a^nb)^{\omega}$ the value is 
$(\frac{(n+1)n}{2} + 1)\cdot \frac{1}{n+1}$, 
and this shows
that the nested weighted automaton can return unbounded values 
(in contrast to a $\flimavg$-automaton whose range is bounded by its maximal weight).
Consider the automaton 
$\nestedA_{\stu}^2 = \langle \masterA^2; \flimavg;  \slaveA_{1}, \slaveA_{2}, \slaveA_3  \rangle$,
where $\slaveA_3$ has only a single state, which is accepting, and it has no transitions.
Thus, $\slaveA_3$ accepts on the empty word  and 
invoking it is a way for $\masterA^2$ to take a silent transition.
Intuitively,  each slave automaton counts how many times a given letter occurs.
Thus the value computed by the nested weighted automaton is the average letter repetition
(or average stuttering).
Silent transitions, produced by calling the automaton $\slaveA_{3}$, enable $\masterA^2$ to compute average only over positions
where a new block starts.
\end{example}
\begin{wrapfigure}{r}{0.5\textwidth}
\begin{tikzpicture}
\begin{scope}[yshift=0.2cm]
\tikzMasterOne
\node at (0.2,1.0) {(a) $\masterA^1$};
\end{scope}

\begin{scope}[xshift=1.0cm,yshift=0.2cm]
\tikzMasterTwo
\node at (0.8,1.0) {(b) $\masterA^2$};
\end{scope}

\begin{scope}[xshift=0.8cm]
\tikzSlaveA
\node at (0.8,-0.6) {(c) $\slaveA_1$};
\end{scope}

\begin{scope}[xshift=3.3cm]
\tikzOfSlaveRun
\node at (2,-0.6) {(d)~A run of $\nestedA_{\stu}^1$};
\end{scope}
\end{tikzpicture}
\caption{The master automata (a)~$\masterA^1$, (b)~$\masterA^2$,
(c)~the slave automaton $\slaveA_1$, (d)~a run of  
the nested weighted automaton $\nestedA_{\stu}^1$.}
\label{fig:run}
\vspace{-10pt}
\end{wrapfigure}

\begin{example}[Average response time]
\label{ex:avg-resp-time}
Consider the specification for average response time defined as follows:
we consider words for the alphabet $\{r,g,i\}$, where $r$ denotes a 
request, $g$ denotes a grant, and $i$ denotes idle (no request or grant).
Consider a word $w$, and a position $j$, such that $w[j]$ is a request,
and then the response time in position $j$ is the distance to the 
closest grant, i.e., the response time is $j'-j$ where $j'>j$ is the least
number greater than $j$ with $w[j']=g$.
The average response time is the limit-average of the response times
of the requests. 
Consider a nested weighted automaton,
with one slave automaton that has sum of non-negative weights as the 
value function, and the master automaton with limit-average value function.
The master automaton for every letter $r$ invokes the slave automaton,
and for $g$ and $i$ takes a silent transition (i.e., it is a single
state automaton with $r$ labeled as~1, and $g$ and $i$ labeled as~2).
The slave automaton $\slaveA_1$ counts the number of steps till the first $g$ and
the slave automaton $\slaveA_2$ accepts only the empty word, which is used to 
produce silent transitions.
The nested weighted automaton specifies the average response time property.
As discussed in Section~\ref{sec:intro} since the average response time 
can be unbounded, it cannot be expressed by a non-nested limit-average
automaton, whose value is bounded by the maximal weight that occurs in it.
\end{example}

\smallskip\noindent{\em Equivalence with weighted automata.} 
We say that a nested weighted automaton $\nestedA$ and a weighted automaton $\aut$
are \emph{equivalent} iff their values coincide on each word,
i.e., for all $w \in \Sigma^{\omega}$ we have $\valueL{\nestedA}(w) = \valueL{\aut}(w)$.

\smallskip\noindent{\em Determinism of nested weighted automata.}
There are two reasons why a nested weighted automaton may be non-deterministic. 
The first one is standard: one of the components, the master automaton or
one of the slave automata is non-deterministic. The second one is more subtle:
it is the termination of slave automata. To accept, a slave automaton
has to terminate in an accepting state, but it not need to be the first
time it visits an accepting state. It can run longer to compute a different value.
However, if the language $\lang$ recognized by the slave automaton is \emph{prefix-free}, i.e.,
 $w \in \lang$ implies that no extension of $w$ belongs to $\lang$,
then it has to terminate once it reaches an accepting state because 
it will have no other chance to accept. This intuition suggests 
the following definition.

\smallskip\noindent{\em Types of nested weighted automata.}
A nested weighted automaton is \emph{deterministic} iff 
the master automaton and all slave automata are deterministic 
and each slave automaton recognizes a prefix-free language.
A nested weighted automaton is \emph{functional} iff 
for every word $w$, each accepting run on $w$ has the same weight.

We will consider the decision questions of emptiness and universality for
nested weighted automata.

\section{Decision Problems}
\label{s:decision}
In this section we study the decidability and complexity of the decision problems
for nested weighted automata.
We start with some simple observations.

\smallskip\noindent{\em Simple observations.} 
Note that the emptiness (resp. universality) of $f$-automata and 
$g$-automata reduces to the emptiness (resp. universality) of 
$(f;g)$-automata: by simply considering dummy master or dummy slave
automata.
Hence by Theorem~\ref{t:oldResults} it follows that the universality 
problem for $(f;g)$-automata is decidable only if the universality problem
is decidable for $f$-automata and $g$-automata.

\begin{restatable}{theorem}{SimpleReduction}
\label{thm:simple}
(1)~For $f \in \InfVal$, the universality problem  for $(f; \fsum)$-automata is undecidable.
(2)~For $g \in \FinVal$, the universality problem  $(\flimavg; g)$-automata is undecidable.
\end{restatable}
\begin{proof}[Proof of (1) from Theorem \ref{thm:simple}]
We show a reduction of the universality problem for $\fsum$-automata with weights $\{ -1, 0, 1\}$, 
which is undecidable (Theorem~\ref{t:oldResults}),
to the universality problem for $(f;\fsum)$-automata, 
where $f \in \{ \finf, \fsup, \fliminf, \flimsup\}$. The case $f  = \flimavg$
follows from (2).

Let $\aut$ be a $\fsum$-automaton with weights $\{ -1, 0, 1\}$. 
Consider an $(\finf;\fsum)$-automaton $\nestedA$
that works as follows. Its acceptance condition enforces that it accepts only 
words with infinitely many $\#$ letters, i.e., 
the words the form $w_1 \# w_2 \# \dots$.
At each $\#$ letter the master automaton starts an instance of $\aut$ as a slave automaton
that works to the successive $\#$ letter. 
On the positions with a letter different than $\#$, the master automaton takes 
a silent transition. 
Then, the value of a word $w_1 \# w_2 \# \dots$ is equal to the infimum
of values $\valueL{\aut}(w_i)$. In particular, $\valueL{\nestedA}((w \#)^{\omega}) = \valueL{\aut}(w)$.
Since for every word $w_1 \# w_2 \# \dots$ there exists $i$ such that
$\valueL{\nestedA}(w_1 \# w_2 \# \dots) = \valueL{\nestedA}((w_i \#)^{\omega})$,
the universality problems for $\nestedA$ and $\aut$ coincide. 

The same construction shows a reduction of the universality problem for 
$\fsum$-automata 
to the universality problem for $(\fliminf;\fsum)$-automata
(resp. $(\fsup;\fsum)$-automata,  $(\flimsup;\fsum)$-automata).
\end{proof}
\begin{proof}[Proof of (2) from Theorem \ref{thm:simple}]
For every $g \in \FinVal$ we can define
two dummy $g$-automata, $\aut_{0}$ (resp. $\aut_{1}$) that immediately accept
and return the value $0$ (resp. $1$). Therefore, 
such $(\flimavg;g)$-automata can simulate all $\flimavg$-automata with weights $0,1$,
whose universality problem is undecidable (Theorem~\ref{t:oldResults}).
Therefore, the universality problem for $(\flimavg;g)$-automata is undecidable as well.
\end{proof}

In the decidable cases, the lower bound for the emptiness and the universality problems 
is $\PSPACE$. Recall, the finite automata intersection problem,
which asks, given a set of deterministic finite automata, is there a finite word accepted by all of them,
is $\PSPACE$-complete~\cite{DBLP:conf/focs/Kozen77}. That problem can be reduced 
to the emptiness problem of deterministic nested weighted automata. 
This implies $\PSPACE$-hardness of the emptiness problem for deterministic nested weighted automata. 
Moreover, 
for deterministic nested weighted automata, the emptiness problem 
can be reduced to the universality problem.
Therefore, the universality problem for deterministic nested weighted automata 
is also $\PSPACE$-hard.

\begin{restatable}{proposition}{PSpaceHardness}
\label{p:pspacehard}
For all $f \in \InfVal$ and $g \in \FinVal$, 
the emptiness (resp. the universality) problem for deterministic nested 
weighted automata is $\PSPACE$-hard.
\end{restatable}
\begin{proof}
Let $\aut_1, \ldots, \aut_n$ be deterministic finite automata over the alphabet $\Sigma$. 
First, we modify each of them to obtain $\aut_1^*, \dots, \aut_n^*$ over $\Sigma \cup \{ \#, \$\}$
such that $\#,\$ \notin \Sigma$ and for every $i \in \{1,\ldots, n\}$ we have
$\aut_i^*$ accepts $u$ iff $u = \#^k w \$$ and $\aut_i$ accepts $w$. 
Observe that  $\aut_1^*, \dots, \aut_n^*$ recognize prefix-free languages.
Then, we define a nested weighted automaton whose slave automata are 
$\aut_1^*, \dots, \aut_n^*$ and the master automaton recognizes the language
$(\#^n \{a,b\}^* \$)^{\omega}$. The master automaton invokes all slave automata
on successive $\#$ letters. Thus, for a word $\#^n w_1 \$ \#^n w_2 \$ \ldots$ to 
be accepted, the slave automaton $\aut_1^*$ has to accept the words 
$\#^n w_1 \$, \#^n w_2 \$, \ldots$, 
 $\aut_2^*$ has to accept the words 
$\#^{n-1} w_1 \$, \#^{n-1} w_2 \$, \ldots$, and so on. 
It follows that the nested weighted automaton accepts the language
$\#^n w_1 \$ \#^n w_2 \$ \ldots$ such that all words $w_1, w_2 $ are accepted by 
all automata $\aut_1, \ldots, \aut_n$. Therefore, the nested weighted automaton accepts any word iff
there exists a finite word accepted by all automata $\aut_1, \ldots, \aut_n$.
The nested automaton defined as above is deterministic.
\end{proof}

\subsection{Regular Weighted Slave Automata}
\label{s:regularslave}
We present a general result that ensures decidability for the 
decision problems for a large class of nested weighted automata.
We now consider slave automata that can only return values from a bounded
domain, and present decidability results for them.

\begin{definition}[Regular weighted automata]
Let $\aut$ be a weighted automaton over finite words.
We say that the weighted automaton $\aut$ is a \emph{regular weighted} automaton iff
there is a finite set $\{ q_1, \ldots, q_n\} \subseteq \Q$ and 
there are regular languages ${\lang}_1, \ldots, {\lang}_n$
such that 
\begin{compactenum}[(i)]
\item every word accepted by $\cal A$ belongs to $\bigcup_{1 \leq i \leq n} {\lang}_i$, and
\item for every $w \in {\lang}_{i}$, each run of $\cal A$ on $w$ 
has the weight $q_i$.
\end{compactenum}
\end{definition}

\begin{remark}
Regular weighted automata define the class of are equivalent to recognizable step functions~\cite{Droste:2009:HWA:1667106}. 
However, we (implicitly) require regular languages ${\lang}_1, \ldots, {\lang}_n$
to be disjoint, whereas the value of a recognizable step function at a word $w$ is defined 
as the minimum $q_i \in \{q_1, \ldots, q_n\}$ among such $i$'s that $w$ belongs to $\lang_i$.
\end{remark}

We define the \emph{description size} of a given regular weighted automaton $\aut$,
as the size of automata $\aut_1, \ldots, \aut_n$ recognizing languages 
 ${\lang}_1, \ldots, {\lang}_n$ that witness $\aut$ being a regular weighted automaton.

\noindent{\em Regular value functions.}
A value function $f$ is a \emph{regular value function}
iff all $f$-automata are regular weighted automata.
Examples of regular value functions are $\fmin,\fmax,\fBsum{B}$.
Observe that the description size of a $\fmin$, $\fmax$ and $\fsum^B$-automaton $\aut$
is polynomial in $|\aut|$, but it is exponential in the length of binary representation of $B$, 
for a $\fsum^B$-automaton $\aut$.

\smallskip\noindent{\em Key reduction lemma.}
In the following key lemma we establish that if the slave automata are
regular weighted automata, then nested weighted automata can be reduced
to weighted automata with the same value function as for the master automata.
For regular weighted slave automata, a weighted automaton can simulate a 
nested weighted automaton in the following way. 
Instead of starting a slave automaton, the weighted automaton guesses the 
weight of the current transition (i.e., the value to be returned of the slave
automaton) and checks that the guessed weight is correct.
The definition of regular weighted automata implies that such a check can be 
done by a (non-weighted) finite automaton $\selectA$.
Thus, the weighted automaton takes a universal transition such that in one 
branch it continues its execution and in another it runs $\selectA$.
Observe that such a universal transition can be removed by a standard
power-set construction.
Given a value function $f$, recall that $\silent{f}$ is the value function 
that applies $f$ on sequences after removing silent transitions.
The following Lemma~\ref{t:regular-non-nested} along with Theorem~\ref{t:oldResults} implies Theorem~\ref{thm:dec1}.

\begin{restatable}[Key reduction lemma]{lemma}{RegularReduceToNonNested}
\label{t:regular-non-nested}
Let $f\in \InfVal$ be a value function.
Consider a nested weighted automaton $\nestedA = \langle \masterA; f; \slaveA_1, \ldots, \slaveA_k \rangle$
such that all automata $\slaveA_1, \ldots, \slaveA_k$ are regular weighted automata.
There is a $\silent{f}$-automaton $\nonnestedA$ (weighted automaton), 
that can be constructed in polynomial space, which is equivalent to 
$\nestedA$;
moreover, if $\nestedA$ is functional, 
then $\nonnestedA$ is functional as well.
\end{restatable}
\begin{proof}
\newcommand{\step}{\textsc{Step}}
Assume that each slave automaton $\slaveA_i$ has the weights 
from the set $\{ -n, \ldots, n\}$. Then, since all of the slave automata
are regular weighted automata, for all $i \in \{1,\dots, k\}$ and $j \in \{-n,\dots, n\}$
there is a deterministic finite word automaton $\selectA_{i,j}$ that 
recognizes the language of all words $w$ such that $\valueL{\slaveA_i}(w) = j$. 
Since $\slaveA_i$ is a regular weighted automaton,
it accepts precisely when one of the automata $\selectA_{i,0}, \dots, \selectA_{i,n}$
accepts. 

We define $Q_S$ (resp. $F_S$) as the disjoint union of the sets of states (resp.
the sets of accepting states states)  of all automata $\selectA_{i,j}$.
Let $Q_m$ (resp. $F_m$) be the set of all states (resp. all accepting states) the master automaton $\masterA$. 
We define a relation ${\step} \subseteq 2^{Q_S}\times \Sigma \times 2^{Q_S}$, which
is the union of transition relations lifted to sets of states, i.e., 
$(\{ q_1, \dots, q_l \}, a, \{ q_1', \dots, q_l' \}) \in \step$ iff
for every $m \in \{1,\ldots, l\}$, some automaton $\selectA_{i,j}$ has
a transition $(q_m, a, q_m')$.

We define $\nonnestedA$, which we show is equivalent to $\nestedA$, as 
a generalized \buchi{} automaton, which differs from an automaton over infinite words (\buchi{} automaton)
in the  acceptance condition. 
An acceptance condition in a generalized \buchi{} automaton is a sequence of $F_1, \ldots, F_s$
of sets of states. A run is accepting iff for each $d \in \{1,\dots,s\}$ there is a state from $F_d$
visited infinitely often.
There is a straightforward reduction of a generalized \buchi{} automaton to a \buchi{} automaton,
and we omit the reduction and for technical convenience consider generalized \buchi{} condition 
for the proof.
 
The automaton $\nonnestedA$ works as follows. It simulates the execution of the master automaton.
Every time the master automaton starts a slave automaton $\slaveA_i$, the automaton
guesses the value $j$ that $\slaveA_i$ returns and checks it, i.e., it starts simulating the automaton $\selectA_{i,j}$, 
by including the  initial state of  $\selectA_{i,j}$ in a set of states $P_1$.
The automaton $\nonnestedA$ maintains two sets of states of simulated automata, $P_1$ and $P_2$:
states in $P_1$ and $P_2$ represent states of $\selectA_{i,j}$ and 
basically, there are states in $P_2$ until all automata corresponding to them terminate.
Once they do, $P_2$ is empty and all states from $P_1$ are copied to $P_2$.
Intuitively, the role of $P_1$ and $P_2$ is to ensure that each automaton terminates,
by enforcing $P_2$ to be empty infinitely often.
We now formally define  
$\nonnestedA=\langle \Sigma, Q, q_0, \cost, \delta, F \rangle$ as follows: 
\begin{enumerate}
\item $Q = Q_m \times (\{-n,\ldots, n\} \cup \{ \bot \}) \times 2^{Q_{\selectA}} \times 2^{Q_{\selectA}}$
\item $q_0 = \quadruple{q^m_0}{0}{\emptyset}{\emptyset}$, where $q^m_0$ is the initial state of the master automaton
\item $(\quadruple{q}{j}{P_1}{P_2}, a, \quadruple{q'}{j'}{P_1'}{P_2'}) \in \delta$ iff 
$(q,a, q')$ is a valid transition of the master automaton labeled by $i$
and one of the following holds (intuitive descriptions follow):
\begin{enumerate}
\item $j = \bot$, $P_1' = P_1'' \setminus F_{S}, P_2' = P_2'' \setminus F_{S}$,
where $\step(P_1,a, P_1'')$ and $\step(P_2,a, P_2'')$,
\item $j \neq \bot$, $P_2 = \emptyset$, $P_1' = \{ q_0^{i,j}\} $  and $P_2' = P_2'' \setminus F_{S}$, where
$\step(P_1,a, P_2'')$ and $q_0^{i}$ is the initial state of $\selectA_{i,j}$,
the automaton that checks that the slave automaton $\slaveA_i$ started at the current position
returns the value $j$, 
\item $j \neq \bot$, $P_2 \neq \emptyset$, 
$P_1' = (P_1'' \cup \{ q_0^{i,j'}\}) \setminus F_S$ and 
$P_2' = P_2''  \setminus F_S$, where
$\step(P_1, a, P_1'')$ and $\step(P_2, a, P_2'')$
\end{enumerate}
The intuitive descriptions are as follows: 
(a)~the first transition corresponds to a silent transition, and hence 
we compute the successor states of sets $P_1$ and $P_2$ and remove the accepting 
states (that correspond to automata that terminate);
(b)~the second transition is similar to the first case but here a new automaton 
that simulates the slave automaton is started, but since $P_2$ is empty we
compute the next $P_2'$ from the successor of $P_1$ according to $\step$ but after
removing the accepting states, and the new $P_1$ is the initial state of the simulating
automaton; and
(c)~the third transition is very similar to the first transition just that the initial
state of the simulating automaton is added to the $P_1'$.

\item the cost function is defined as 
$\cost(\quadruple{q}{j}{P_1}{P_2}, a, \quadruple{q'}{j'}{P_1'}{P_2'}) = j'$,
\item $F$ consists of $F_1= F_m \times (\{-n,\ldots, n\} \cup \{ \bot \}) \times 2^{Q_{S}} \times 2^{Q_{S}}$ 
and $F_2=Q_m \times (\{-n,\ldots, n\} \cup \{ \bot \}) \times 2^{Q_{S}} \times \emptyset$.
Intuitively, $F_1$ ensures that the acceptance condition of the master automaton is satisfied
and $F_2$ ensures that $P_2$ is empty infinitely often.
\end{enumerate}
The correctness follows from the construction.
\end{proof}

The automaton $\nonnestedA$ in Lemma~\ref{t:regular-non-nested} is constructed in polynomial space, which
means that $\nonnestedA$ can be represented implicitly, i.e., its exponential-size 
set of states is represented in a compact way and for each transition triple 
$(q,a,q')$ one can compute in polynomial time whether that triple is a transition of $\nonnestedA$
and what is its weight.

\begin{restatable}{remark}{FirstRemarkParametrizedComplexity}
\label{rem:firstParRemark}
The automaton $\nonnestedA$ from Lemma~\ref{t:regular-non-nested} has exponential size in $|\nestedA|$.
More precisely, the size of $\nonnestedA$ is exponential in the total size of slave automata of $\nestedA$, but only polynomial in
the size of  the master automaton of $\nestedA$.
\end{restatable}
\begin{proof}
The set of states of the automaton $\nonnestedA$ from 
is $Q = Q_m \times (\{-n,\ldots, n\} \cup \{ \bot \}) \times 2^{Q_{\selectA}} \times 2^{Q_{\selectA}}$, i.e.,
it is linear in the size of the master automaton of $\nestedA$. Moreover, the weights of $\nonnestedA$ are bounded by a constant $n$.
Thus, $\nonnestedA$ is polynomial in the size of the master automaton of $\nestedA$.
\end{proof}

Now, we show a simple lemma regarding weighted automata with silent moves.

\begin{restatable}{lemma}{SilSelective}
\label{silnce-simple-equivalence}
Let $f \in \{ \finf,\fsup,\fliminf,\flimsup\}$.
(1)~The emptiness problem for $\silent{f}$-automata is in $\NLOGSPACE$.
(2)~The universality problem for $\silent{f}$-automata is in $\PSPACE$.
\end{restatable}
\begin{proof}
Given a $\silent{f}$-automaton $\aut$, where $f \in \InfVal$, 
we define the automaton $\aut^{\ell}$ as the $f$-automaton that results
from $\aut$ by substituting each silent transition by a transition 
of the weight $\ell$.
Observe that for every $\silent{\finf}$-automaton $\aut$
for every infinite word $w$ we have $\valueL{\aut}(w) \leq \lambda$
iff $\valueL{\aut^{(\lambda+1)}}(w) \leq \lambda$. 
The same equivalence holds for every $\silent{\fsup}$-automaton $\aut$
and its variant ${\aut^{(\lambda-1)}}$. 
Thus, the emptiness and universality problems
for $\silent{\finf}$-automata (resp. $\silent{\fsup}$-automata)  and $\finf$-automata 
(resp. ${\fsup}$-automata) coincide. 
Now, a run of a $\silent{\finf}$-automaton is accepting only if it contains
infinitely many non-silent transitions. Therefore, the above equivalences hold 
for $f \in \{ \fliminf, \flimsup\}$ and the corresponding problems coincide.
As the emptiness problem for ${f}$-automata is in $\NLOGSPACE$ we have (1).
The universality problem for ${f}$-automata is in $\PSPACE$, hence
 we have (2).
\end{proof}

Finally, we are ready to prove theorem characterizing complexity of decision
problem for newsted weighted automata whose slave automata are 
$\{\fmin,\fmax,\fBsum{B}\}$.

\begin{restatable}{theorem}{KeyLemmaConsequences}
\label{thm:dec1}
Let $g \in \{\fmin,\fmax,\fBsum{B}\}$. The following assertions hold: 
(1)~Let $f \in \{\finf,\fsup,\fliminf,\flimsup\}$.
The emptiness problem for non-deterministic $(f;g)$-automata is $\PSPACE$-complete.
The universality problem for non-deterministic $(f;g)$-automata is $\PSPACE$-hard 
and in $\EXPSPACE$.
(2)~The emptiness problem for non-deterministic $(\flimavg;g)$-automata is $\PSPACE$-complete
(3)~The universality problem for functional $(\flimavg;g)$-automata is $\PSPACE$-complete.
\end{restatable}
\begin{proof}
$\PSPACE$-hardness in (1), (2) and (3) follows from Proposition~\ref{p:pspacehard}.
We will discuss containment separately:

(1): Let  $f \in \{\finf,\fsup,\fliminf,\flimsup\}$ and 
$g \in \{\fmin,\fmax,\fBsum{B}\}$.
Due to Lemma~\ref{t:regular-non-nested} 
every $(f;g)$-automaton $\aut$ is equivalent to some $\silent{f}$-automaton $\aut'$ of exponential
size in $|\aut|$. By Lemma~\ref{silnce-simple-equivalence},
the emptiness problem for $\silent{f}$-automata is  in $\NLOGSPACE$.
The construction from Lemma~\ref{t:regular-non-nested} 
implies that the automaton $\aut'$ can be represented implicitly, i.e.,
given two states $q,q'$ the existence and weight of the transition 
$(q,a,q')$ can be decided in polynomial time. Therefore,
the emptiness problem for $(f;g)$-automata is  in $\PSPACE$.

By Lemma~\ref{silnce-simple-equivalence}, 
and the universality problem for $\silent{f}$-automata 
is in $\PSPACE$ and $|\aut'|$ is of exponential size in $|\aut|$, 
hence we have the universality problem for $(f;g)$-automata is in $\EXPSPACE$.

(2): Lemma~\ref{t:regular-non-nested} state that $(\flimavg;g)$-automata
are equivalent to $\silent{\flimavg}$-automata, which
enjoy decidability of the emptiness problem (Lemma~\ref{silent-lim-avg}) in $\NLOGSPACE$.
As in \emph{(1)}, the automaton $\aut'$ can be represented implicitly, hence
the emptiness problem for $(\flimavg;g)$-automata is in $\PSPACE$.

(3): The universality problem for functional $(\flimavg;g)$-automata reduces
to the emptiness problem for functional $(\flimavg;g)$-automata. It suffices 
to  (1)~first check that every word has an accepting run, which can be done in polynomial space,
(2)~construct a $(\flimavg;g)$-automaton $\nestedA'$ by 
taking additive inverses of all weights in all slave automata of a given nested weighted automaton.
The automaton $\nestedA$ satisfies the universality problem with threshold $\lambda$
iff it satisfies (1) and the automaton $\nestedA'$ from (2) does not satisfy the emptiness problem
with threshold $-\lambda$.
Therefore, the universality problem for functional $(\flimavg;g)$-automata
is in $\PSPACE$.
\end{proof}

\begin{remark}
Assume that the total size of slave automata is bounded. Then, by Remark~\ref{rem:firstParRemark}, 
the size of the automaton $\nonnestedA$ is polynomial in the size of the master automaton of a given nested automaton.
In consequence, the emptiness problem for automata from (1) and (2) from Theorem~\ref{thm:dec1} becomes $\PTIME$
and the universality problem  for automata from (1) from Theorem~\ref{thm:dec1} becomes $\PSPACE$-complete. 
\end{remark}

Theorem~\ref{thm:dec1} covers the case for all classes of slave automata other 
than $\fsum$- and $\fsum^+$-automata, which we consider in the following
two subsections.

\subsection{Undecidability Results for Slave $\fsum$ Automata}
\label{s:undecidableNested}
In this section we study $(f; \fsum)$-automata and
we present a crucial negative result. 

Note that for weighted automata 
with the value function from $\FinVal$ or $\InfVal$, 
 the emptiness problem is decidable (for non-deterministic automata);
and all decision problems are decidable for deterministic automata.
In sharp contrast we establish that for deterministic $(\fsup;\fsum)$-automata 
the emptiness problem is undecidable.
The proof is a reduction from the halting problem of a two-counter (Minsky)
machine to the emptiness problem.
The key idea is to ensure that words that encode valid computations of the 
Minsky machine have value~0; and all invalid computations have value 
strictly greater than~0.
Basically, we need to check consistency of values of each counter at each step, which is done as follows.
The task of the master automaton is to ensure that tests on the counters $c_0, c_1$
are consistent.
The master automaton uses several slave automata to track the exact values of
the counters.
Each slave automaton operates on an alphabet which is increment and decrement
for the counters $c_0, c_1$, as well as zero and positive test, and for each counter we have 
three slave automata.
For positions $i<j$, let \emph{$c_0$-balance} (resp. $c_1$-balance) between position $i$ and $j$ denote
the difference in the number of increments and decrements of the counter $c_0$ (resp. $c_1$) between $i$ and $j$.
For zero tests of a counter, two slave automata are invoked: 
the first automaton (resp. second automaton) increments (resp. decrements) 
with every increment operation on the counter and decrements (resp. increments) 
with every decrement  operation on the counter and terminates with the value at 
the position of the next zero test.
Intuitively, the two automata compute $c_0$-balance and the opposite (the additive inverse) of $c_0$-balance between
two consecutive zero tests.
Given the zero test of the current position is satisfied, both automata return
zero iff the next zero test is also satisfied, otherwise one of them return a 
positive value.
For positive tests of a counter we use the third slave automaton to compute 
the $c_0$-balance plus~1 between the current position and the next zero test of $c_0$.
The $c_0$-balance plus~1 does not exceed zero iff the value of $c_0$ at the current position is positive.
We repeat a similar construction for $c_1$.
The construction of slave automata does not depend on the given two-counter machine,
therefore the reduction works even in the presence of a constant bound on the size of slave automata.

This establishes the undecidability for emptiness of $(\fsup;\fsum)$-automata,
and the proof also holds for $(\flimsup;\fsum)$-automata.
Also observe that since we establish the result for deterministic automata, 
we can take opposites of weights and change $\fsup$ (resp. $\flimsup$) to 
$\finf$ (resp. $\fliminf$) and the emptiness problem to the universality problem.

\begin{restatable}[Crucial undecidability result]
{theorem}{SupSumUndecidable}
\label{thm:crucial-undec}
(1)~The emptiness problem for deterministic $(\fsup;\fsum)$- and $(\flimsup;\fsum)$-automata 
is undecidable.
(2)~The universality problem for deterministic $(\finf;\fsum)$- and $(\fliminf;\fsum)$-automata
is undecidable.
\end{restatable}
\begin{proof}[Proof of (1) from Theorem \ref{thm:crucial-undec}]
\newcommand{\machine}{{\cal M}}
\newcommand{\zeroTestOne}{\mathrm{c_1=0}}
\newcommand{\zeroTestTwo}{\mathrm{c_2=0}}
\newcommand{\posTestOne}{\mathrm{c_1>0}}
\newcommand{\posTestTwo}{\mathrm{c_2>0}}
Given a Minsky machine $\machine$, we construct a deterministic $(\fsup;\fsum)$-automaton $\nestedA$ 
that accepts infinite words of the form 
$w_1 \# w_2 \# \ldots $. Moreover, the value of the word $w_1 \# w_2 \# \ldots $ is $0$
iff each subword $w_i$ encodes a valid accepting computation of $\machine$.
As the problem, given a Minsky machine, does it have an accepting computation is undecidable,
we conclude that the emptiness problem for deterministic $(\fsup;\fsum)$-automata (resp. $(\flimsup;\fsum)$-automata)
is undecidable.

A Minsky machine $\machine$ is a finite automaton augmented with two counters $c_1, c_2$.
The counters can be incremented, decremented and tested  whether they are zero 
or positive. 
The transitions of $\machine$ depend on the values of counters, namely,
whether they are equal zero. That is,
each transition has the following form $(q, s, t) \rightarrow (q', v_1, v_2)$,
where $s \in \{ \zeroTestOne, \posTestOne\}, t \in \{ \zeroTestTwo, \posTestTwo\}$ and $v_1, v_2 \in \{ -1, 0,1\}$.
E.g. $(q,\zeroTestOne,\posTestTwo) \rightarrow (q', +1, -1)$
means that if the machine is in the state $q$, the value 
of $c_1$ is $0$ and $c_2$ greater than $0$,
then the next state is $q'$, $c_1$ is incremented and $c_2$ is decremented.

We define two notions for Minsky machines, a \emph{run} and a \emph{computation}. 
A \emph{run} of a Minsky machine $\machine$ is a sequence 
$(q_0, 0,0), (q_1, \alpha_1, \beta_1), \ldots, (q_n, \alpha_n, \beta_n)$ such that
for every $i < n$ there is a transition of $\machine$
$(q_i, s, t) \rightarrow (q_{i+1}, v_1, v_2)$ such that
$\alpha_i$ satisfies $s$, $\beta_i$ satisfies $t$, 
and $\alpha_{i+1} = \alpha_i + v_1$, $\beta_{i+1} = \beta_i + v_2$.
A run is \emph{accepting} iff its last element is $(q_F, 0,0)$.
A \emph{computation} of $\machine$ is a sequence of elements
$Q \times \{ \zeroTestOne, \posTestOne\} \times \{ \zeroTestTwo, \posTestTwo\} \times \{ -1,0,1\} \times \{ -1,0,1\}$ called \emph{configurations}.
A computation 
$(q_0, \zeroTestOne, \zeroTestTwo, 0, 0),$ $(q_1, s_1, t_1, x_1, y_1),$ $\ldots,$ $(q_n, \zeroTestOne, \zeroTestTwo, x_n, y_n)$
is \emph{valid} iff there is an accepting run 
$(q_0, 0,0), \ldots, (q_n, \alpha_n,\beta_n)$
such that for every $i \in \{0,\ldots, n\}$,
$\alpha_i = \sum_{j=0}^{i} x_j$ and 
$\beta_i = \sum_{j=0}^{i} y_j$.

Consider a valid computation $\eta$ and the corresponding accepting run $\pi$.
For positions $i<j$, let \emph{$c_1$-balance} (resp. $c_2$)-balance between position $i$ and $j$ (in $\eta$) denote
the difference in the number of increments and decrements of $c_1$ (resp. $c_2$) between $i$ and $j$.
Since the initial value of the counters is $0$, the value of a counter $c_1$ (resp. $c_2$) in $\pi[i]$
is precisely its $c_1$-balance (resp. $c_2$-balance) between positions $1$ and $i$. 
Thus, for $p \in \{1,2\}$, a zero test (non-zero test) of $c_p$ at the position $i$ is valid iff $c_p$-balance between positions 
$1$ and $i$ is $0$ (is strictly positive).

Consider a computation $\eta$ of a Minsky machine $\machine$.
If it is invalid then there is a first position in $\eta$ such that the corresponding sequence
over $Q \times \N \times \N$ is not a run. 
There are two possible reasons for that:
(i) $\machine$ has no transition consistent with a step from $\eta[i]$ to $\eta[i+1]$ ,
(ii) the configuration at $\eta[i]$ is inconsistent with the current values of $c_1, c_2$,
i.e., a zero or a non-zero test is inconsistent with the actual value of a counter.
A Boolean automaton can check whether the computation is invalid because of (i).
We show how to check (ii), i.e., validity of zero and non-zero tests, 
 using a nested weighted automaton.  

Let $p \in \{1,2\}$.
First, we check validity of zero tests on $c_p$. All zero tests on $c_p$ are valid iff
$c_p$-balance between any two consecutive zero tests is zero.
To check that this holds, 
the nested weighted automaton starts at each position $i$ with a zero test two deterministic slave $\fsum$-automata: ${\cal A}_{\zeroTestOne}^{+}, {\cal A}_{\zeroTestOne}^{-}$. 
The automaton ${\cal A}_{\zeroTestOne}^{+}$ computes $c_p$-balance between $i$ and the next zero test of $c_p$;
it increments (decrements) its value whenever 
$c_p$ is incremented (decremented), and it terminates at the next zero test of $c_p$. 
The automaton ${\cal A}_{\zeroTestOne}^{-}$ does the opposite,
i.e., it computes the additive inverse of $c_p$-balance between $i$ and the next zero test of $c_p$.
The values of these automata are inverses of each other and 
the maximum of their values is the absolute value of $c_p$-balance.
Hence, the maximum of their values is less-or-equal to zero iff $c_p$-balance 
between $i$ and the next zero test of $c_p$ is $0$.
Thus, the values of all slave automata ${\cal A}_{\zeroTestOne}^{+}, {\cal A}_{\zeroTestOne}^{-}$
are less-or-equal to zero if and only if all zero tests of $c_p$ are valid.

Second, we check that non-zero tests are valid. To do that, the automaton starts at every position $i$ with a
non-zero test a third slave $\fsum$-automaton ${\cal A}_{\posTestOne}$ that first increments its value to $1$ and then
computes $c_p$-balance between $i$ and the next zero test of $c_p$.
The value of $c_p$ at the position $i$ is strictly greater than $0$ iff
$c_p$-balance between the position $i$ and the next position at which $c_p$ is $0$
does not exceed $-1$. 
Provided that verifying zero tests succeeds, 
the value of ${\cal A}_{\posTestOne}$ is less-or-equal to $0$
iff the non-zero test at the position $i$ is valid.

The value of the nested weighted automaton does not exceed $0$ if and only if
the values of all slave automata are less-or-equal to $0$, which holds precisely
when all zero and non-zero tests on $c_p$ are valid.
In the above construction up to four automata has to be started
at any configuration, while nested weighted automata can start at most one slave automaton at each step.
However, we can  encode configurations by some fixed number of letters.
E.g. $c \ \$\ \$ \  \$	\ \$ $ where $c$ is a letter that fully encodes a configuration
$(q, \alpha, \beta, x, y)$ and $\$ $ letters are used only to start enough slave automata.
It follows that $\nestedA$ accepts a word
$w_1 \# w_2 \# \ldots $ and assigns it
the value $0$ iff each word $w_i$ encodes an valid accepting computation of $\machine$.

Observe that the same automaton, $\nestedA$, considered as $(\flimsup;\fsum)$-automaton
returns the same result. Indeed, if a given Minsky machine does not have an accepting computation,
each accepted word will have positive value. On the other hand, 
if there is an accepting computation $w$, 
the value of $(\fsup;\fsum)$-automata and  $(\flimsup;\fsum)$-automata
the word $(w \#)^{\omega}$ coincides, hence it is $0$.
\end{proof}

\begin{proof}[Proof of (2) from Theorem \ref{thm:crucial-undec}]
The universality problem for deterministic $(\finf;\fsum)$-automata is the dual of 
the emptiness problem for deterministic $(\fsup;\fsum)$-automata. 
Indeed, consider a deterministic $(\finf;\fsum)$-automaton $\nestedA$ 
and the nested weighted automaton $\nestedA'$ that results from taking inverses
of all weights in $\nestedA$ and changing its value function to $\finf$. 
One can easily check that for every word $w$, 
the weight of $w$ assigned by $\nestedA$ is $x$, then
$\nestedA'$ assigns to $w$ the weight $-x$.
\end{proof}

\subsection{Decidability Results for Slave $\fsum$- and $\fsum^+$-Automata}
\label{s:remainingAutomata}
We now establish the remaining decidability results,
namely, for slave automata with $\fsum^+$ value function, and
emptiness for $(\finf;\fsum)$-automata and $(\fliminf;\fsum)$-automata.
In contrast to the reduction of Lemma~\ref{t:regular-non-nested}, for example, 
$(\flimavg;\fsum^+)$-automaton cannot be reduced to weighted $\silent{\flimavg}$-automata 
(Example~\ref{ex:avg-resp-time}).

\smallskip\noindent{\em Intuitive proof ideas.}
For $(f;\fsum^+)$-automata, for $f \in \InfVal \setminus \{\flimavg\}$, 
we show that the decision problems can be reduced to the bounded sum value function;
and then derive the decidability results from Theorem~\ref{thm:dec1}.
The reductions are polynomial in the size of the master automaton.
For $(\finf;\fsum)$-automata we show the emptiness problem is 
decidable and the main argument is a reduction to the emptiness 
of non-deterministic weighted automata with $\fsum$ value function.
The constructed automaton is exponential in the size of a nested automaton, but only polynomial
in the size of the master automaton, i.e., if the total size of slave automata and the threshold are bounded by a constant.
We summarize the results in the following theorem.

\begin{restatable}{theorem}{Additional} 
(1)~For $f \in \{ \finf, \fliminf \}$, the emptiness problem for $(f;\fsum)$-automata is $\PSPACE$-complete.
(2)~For $f \in \{ \fsup, \flimsup \}$, the universality problem for functional $(f;\fsum)$-automata is $\PSPACE$-complete.
(3)~For $f \in \{ \finf,\fsup,\fliminf,\flimsup\}$, 
the emptiness problem for $(f;\fsum^+)$-automata is $\PSPACE$-complete, and
the universality problem for $(f;\fsum^+)$-automata is $\PSPACE$-hard and in 
$\EXPSPACE$.
\label{thm:additional}
\end{restatable}
\begin{proof}[Proof of (1) from Theorem \ref{thm:additional}]
$\PSPACE$-hardness follows from Proposition~\ref{p:pspacehard}.
For containment in $\PSPACE$, let $\nestedA = \langle \masterA; \finf;  \slaveA_1, \ldots, \slaveA_k \rangle$
be a nested weighted automaton  $(\finf;\fsum)$-automaton.
We construct a $\fsum$-automaton over finite words $\nonnestedA$ 
such that the emptiness problem for $\nestedA$ and $\nonnestedA$ coincide.
The automaton $\nonnestedA$ works over words over the alphabet $\Sigma \cup \{ \#, 1, \ldots k\}$
of the form $w i v \# u' \# u$, where $w,v,u',u \in \Sigma^*$ and $i \in \{1, \ldots k\}$, and the value of its run, if it is accepting,
 is the value of the slave automaton $\slaveA_i$ on the word $v$.
The automaton $\nonnestedA$ consists of two components.
The first component $\aut_1$, a Boolean one whose all weights are $0$, ensures that $\nestedA$ has an accepting run 
on $w v u' u^{\omega}$ such that the slave automaton started at the beginning of the word
$v$ is $\slaveA_i$ and $\slaveA_i$ accepts the word $v$.
The second component, $\aut_2$, is a weighted one and it computes the value of $\slaveA_i$
on $v$. Clearly, the size of $\aut_2$ is proportional to the size of $\slaveA_i$.
Observe that the value of each run of $\nestedA$ depends only on a finite prefix of a word, i.e.,
for each run of $\nestedA$ there is a finite prefix $wvu'u$
such that the value of that run equals $\valueL{\aut}(w i v \# u' \# u)$. 
It follows that the emptiness problem for $\nestedA$ and $\nonnestedA$ coincide.
The construction of $\aut_1$ is similar to the construction from Lemma~\ref{t:regular-non-nested}, hence
the emptiness problem for $\nonnestedA = \aut_1 \times \aut_2$ can be solved in 
polynomial space w.r.t. $|\nestedA|$.

Assume that $\nestedA$ is a  $(\fliminf;\fsum)$-automaton. 
We carry out virtually the same construction
of a $\fsum$-automaton over finite words $\nonnestedA$.
The automaton $\nonnestedA$ accepts words $w i v \# u$ such that 
$\slaveA_i$ accepts $v$ and $\nestedA$ has an accepting run on $w (vu)^{\omega}$ 
at which the slave automaton 
invoked at the positions $w, wvu, \dots, w(vu)^k,\dots$ is $\slaveA_i$.
The value of $\nonnestedA$ on an accepted word $w i v \# u$ is
the value of $\slaveA_i$ on $v$. 
It follows that if $\nonnestedA$ has a run of value $\lambda$ on 
$w i v \# u$, $\nestedA$ has a run of the value $\lambda$ on $w (vu)^{\omega}$.
Conversely, if $\nestedA$ has a run of the value $\lambda$, 
there is a reachable state $q$ of the master automaton $\masterA$ of $\nestedA$ and
a slave automaton $\slaveA_i$ such that infinitely often 
$\masterA$ in the state $q$ invokes $\slaveA_i$ which returns the value $\lambda$.
Thus, there are words $v,u$ such that $\slaveA_i$ on $v$ returns the value $\lambda$
and $\masterA$ upon reading $vu$ returns to the state $q$.
Moreover, there is a word $w$ such that $\masterA$ reaches $q$ from the initial
state upon reading $w$. 
Therefore, the value of $w (vu)^{\omega}$ in $\nestedA$ is at most $\lambda$.
Hence, the emptiness problems for $\nestedA$ and $\nonnestedA$ coincide.
Similarly to the $(\finf,\fsum^+)$ case, the emptiness problem for
$\nonnestedA$  can be solved in polynomial space w.r.t. $|\nestedA|$.
\end{proof}

\begin{remark}
The construction of $\aut_1$ is similar to the construction from Lemma~\ref{t:regular-non-nested}, hence it is polynomial
in the size of the master automaton. 
Therefore, the emptiness problem for  $(\finf;\fsum)$-automata (resp.  $(\fliminf;\fsum)$-automata) is in $\PTIME$ provided that the 
total size of slave automata is bounded.
\end{remark}

\begin{proof}[Proof of (2) from Theorem \ref{thm:additional}]
$\PSPACE$-hardness follows from Proposition~\ref{p:pspacehard}.
The universality problem for functional $(\finf;\fsum)$-automata (resp. $(\fliminf;\fsum)$-automata) reduces
to the emptiness problem for functional $(\fsup;\fsum)$-automata (resp. $(\flimsup;\fsum)$-automata). It suffices 
to  (1)~first check that every word has an accepting run, which can be done in polynomial space,
(2)~construct an automaton $(\finf;\fsum)$-automaton (resp. $(\fliminf;\fsum)$-automaton) $\nestedA'$ by 
taking inverses of all weights in all slave automata of a given nested weighted automaton.
The automaton $\nestedA$ satisfies the universality problem with threshold $\lambda$
iff it satisfies (1) and the automaton $\nestedA'$ from (2) does not satisfy the emptiness problem
with threshold $-\lambda$.
Therefore, the universality problem for functional $(\finf;\fsum)$-automata (resp. $(\fliminf;\fsum)$-automata)
is in $\PSPACE$.
\end{proof}

\begin{proof}[Proof of (3) from Theorem \ref{thm:additional}]
$\PSPACE$-hardness follows from Proposition~\ref{p:pspacehard}.

Let $\const$ be the threshold given in the emptiness (resp. universality) problem. 
Consider a $(f; \fBsum{B})$-automaton $\nestedA^{\const}$, where $B=\const+1$, obtained from 
$\nestedA$ by changing each slave $\fsum^+$ automaton $\slaveA$ into $\fBsum{\const+1}$-automaton $\slaveA^{\const}$.
Basically, such a $\fBsum{\const+1}$-automaton $\slaveA^{\const}$ simulates runs of $\fsum^+$-automata by 
implementing a $\const+1$-bounded counter in its states $Q \times \{0,\ldots, \const+1\}$, where $Q$ is the set of states of $\slaveA$. 
If  $\slaveA$ accumulates the value above $\const$, the automaton $\slaveA^{\const}$ returns just ${\const}+1$, 
regardless of the actual value accumulated by $\slaveA$.
The automaton $\nestedA^{\const}$ is polynomial in $\const$, which can be exponential in the input size.
Observe that for $f \in \{ \finf, \fsup, \fliminf, \flimsup \}$, 
for every word $w$, $\nestedA$ has a run on $w$ of the value not exceeding $\const$
threshold iff  $\nestedA^{\const}$ has.
It follows that the emptiness (resp. universality) problem for $(f; \fsum^+)$-automata 
with threshold $\const$ reduces to the emptiness (resp. universality) problem for $(f; \fBsum{\const+1})$-automata.
Since $\fBsum{B}$ is a regular value function, 
Lemma~\ref{t:regular-non-nested} states that
for $f \in \{ \finf, \fsup, \fliminf, \flimsup \}$, 
a $(f; \fBsum{\const+1})$-automaton $\nestedA^{\const}$ is equivalent to a $\silent{f}$-automaton $\nonnestedA$.
Therefore, the emptiness (resp. the universality) problem for $(f; \fsum^+)$-automata reduces
to the emptiness (the universality) problem for $\silent{f}$-automata.
However, by employing Lemma~\ref{t:regular-non-nested}, we get $\nonnestedA$ of 
the size exponential in $|\nestedA^{\const}|$ and doubly-exponential in $|\nestedA|$.
We show that the second exponential blow-up can be avoided. 

We show that there exists a $\silent{f}$-automaton $\nonnestedA^-$, equivalent to
$\nonnestedA$ of the exponential size in the input size.

\noindent{\em Infimum case.} Let $f \in \{ \finf, \fliminf \}$.
Original $\silent{f}$-automaton $\nonnestedA$ simulates runs of all slave automata of $\nestedA^{\const}$. 
The modified $\silent{f}$-automaton $\nonnestedA^-$ simulates only a single $\fBsum{\const+1}$-automaton
at the time, which is chosen non-deterministically. For remaining slave automata,
only their non-weighted counterparts are simulated, i.e., $\fsum^+$ automata from $\nestedA$
with weights removed.
Since $f$ is infimum of limit-infimum value function, the automata 
$\nonnestedA$ and $\nonnestedA^-$ are equivalent.
The cardinality of the set of states of $\nonnestedA^-$ is $O(2^{|\nestedA|} \cdot |\nestedA| \cdot B)$.
Therefore, the size of $\nonnestedA^-$ is exponential in the input size. 

\noindent{\em Supremum case.} Let $f \in  \{ \fsup, \flimsup \}$.
Recall that the set of states of $\slaveA^{\const}$ is $Q \times \{0,\ldots, \const+1\}$, 
where $Q$ is the set of states of $\slaveA$. 
We obtain $\nonnestedA^-$ from $\nonnestedA$, by imposing the following condition: 
(*)~at every position $k$, if $\nonnestedA^-$ simulates two runs $\slaveRun_i, \slaveRun_j$
of $\slaveA^{\const}$ that have states $(q, w_1)$ resp. $(q,w_2)$ at position $k$, with $w_1 > w_2$, 
$\nonnestedA^-$ discards the run $\slaveRun_j$ (the one that has the state $(q, w_2)$). 
Intuitively, the run $\slaveRun_j$ can be completed to an accepting run that 
accumulates lower value than $\slaveRun_i$, thus simulating it is redundant. 
We argue that $\nonnestedA$ and $\nonnestedA^-$ are equivalent. 

The $\nonnestedA^-$ simulates only a subset of slave automata.
Since its value function is $\silent{\fsup}$ or $\silent{\flimsup}$,
for every word $w$, the value of $\nonnestedA^-$ does not exceed the value of $\nonnestedA$.
Conversely, 
consider an accepting run $(\masterRun, \slaveRun_1, \slaveRun_2, \ldots)$ of $\nestedA^{\const}$ on $w$.
We can modify runs of slave automata $\slaveRun_1, \slaveRun_2, \ldots$ so that the modified run $(\masterRun, \slaveRun_1', \slaveRun_2', \ldots)$ 
satisfies the following condition~(**):~at every position $k$ in $w$, if runs $\slaveRun_i,\slaveRun_j$
have states $(q, w_1)$, resp. $(q,w_2)$ at the positions corresponding to $k$,
then they accumulate the same value in the remaining part of of the run.
One can achieve that by changing the suffix of the run that accumulates 
greater value to the suffix of the other run. 
Such an operation of substituting a prefix decreases the value, hence it can be 
executed finitely many times for each run, and it will not produce infinite runs of slave automata.
Observe that the modified run is an accepting run of $\nestedA^{\const}$
of the value not exceeding the value of $(\masterRun, \slaveRun_1, \slaveRun_2, \ldots)$.

Now, observe that for a run of $\nestedA^{\const}$ satisfying (**),
if runs $\slaveRun_i,\slaveRun_j$
have states $(q, w_1)$, resp. $(q,w_2)$ at the positions $k$ in $w$,
with $w_1 > w_2$, the value of $\slaveRun_i$ is greater than the value of $\slaveRun_j$, and
the run $\slaveRun_j$ can be discarded. 
Such an operation corresponds to the condition (*) imposed by $\nonnestedA^-$.
Therefore, the values of $w$ assigned by $\nestedA^{\const}, \nonnestedA$ and
$\nonnestedA^-$ are equal.

The cardinality of the set of states of $\nonnestedA^-$ is $O((|\nestedA|  \cdot B)^{|\nestedA|})$, which is 
exponential in the input size. 

The emptiness (resp. the universality) problem of a $(f; \fsum^+)$-automaton $\nestedA$ reduces
to the emptiness (the universality) problem for $\silent{f}$-automaton $\nonnestedA^-$ of the exponential  
size in $|\nestedA| + \log(\const)$. Hence,
by Lemma~\ref{silnce-simple-equivalence}, for $(f; \fsum^+)$-automata,
the emptiness problem is  in $\PSPACE$ and
the universality problem is in $\EXPSPACE$.
\end{proof}

\begin{remark}
Let $f \in \{ \finf, \fliminf, \fsup, \flimsup \}$.
Assume that the threshold $\lambda$ is given in unary.
Then, for $(f; \fsum^+)$-automata, the emptiness problem is in $\PTIME$ and 
the universality problem is $\PSPACE$-complete.
\end{remark}
\begin{proof}
Let $f \in \{ \finf, \fliminf, \fsup, \flimsup \}$.
Assuming that the threshold is given in unary and the total size of slave automata
is bounded, the size of $\nonnestedA^-$ is polynomial in the size of $\nestedA$. 
Therefore,  the emptiness (resp., the universality) problem for $(f; \fsum^+)$-automata
reduce to the emptiness (resp., the universality) problem for $f$-automata.
The emptiness problem for $f$-automata is in $\PTIME$ and the universality 
problem is $\PSPACE$-complete. Hence, the result follows.
\end{proof}

Finally, we establish decidability of the emptiness problem with limit-average 
master automaton and $\fsum^+$-automata as slave automata.
The key proof idea is to show that values of certain runs of $(\flimavg;\fsum^+)$-automata
coincide with the values of non-nested limit-average automata, and those runs
have values arbitrarily close to the infimum over values of all runs.
This also allows us to show the decidability of the universality problem for functional 
$(\flimavg;\fsum^+)$-automata.

\begin{restatable}{theorem}{LimAvgSum}
The emptiness problem for $(\flimavg;\fsum^+)$-automata 
is $\PSPACE$-hard and in $\EXPSPACE$; 
and the universality problem for functional $(\flimavg;\fsum^+)$-automata
is $\PSPACE$-hard and in $\EXPSPACE$.
\label{thm:limavgsum}
\end{restatable}

We present the proof of part~(1) from Theorem~\ref{thm:crucial-undec} in Section~\ref{s:proof-limavg}.
In the following, we show the proof of part~(2) from Theorem~\ref{thm:crucial-undec}.

Observe that for a run of a functional nested weighted automaton, as long as the run is accepting, its value 
does not depend on the choices of transitions.
Therefore, we will focus on the construction of an accepting run with the maximal value to
compute the minimal threshold for the functionality problem.

\newcommand{\boundOnWeights}{\ensuremath{\Lambda}}

\begin{lemma}
Let $\nestedA$ be a functional $(\flimavg; \fsum^+)$-automaton and 
let $\boundOnWeights$ be the value bounding weights in all slave
automata of $\nestedA$.
Then, one of the following holds:
\begin{enumerate}
\item For every accepting run, there is a position $s_0$ such that
every slave automaton started after $s_0$ accumulates the value
not exceeding $\boundOnWeights \cdot \confNum{\nestedA}$.
\item The automaton $\nestedA$ has an accepting run of infinite value 
(whose value exceeds every $\lambda > 0$).
\end{enumerate}
\end{lemma}
\begin{proof}
Assume that (1) does not hold.  Then, there is an accepting run such that some slave automaton 
returns values that exceed the value $\boundOnWeights \cdot \confNum{\nestedA}$ infinitely often. 
Observe that if a slave automaton $\slaveA$ accumulates a value exceeding $\boundOnWeights \cdot \confNum{\nestedA}$ 
during a run $\pi$, then the nested weighted automaton $\nestedA$ is in the same configuration at least twice
during the run $\pi$ and meanwhile $\slaveA$ increases its value. 
Therefore, one can pump the run of the nested weighted automaton to increase the value returned by 
$\slaveA$.
It follows that we can pump successively the run on $\nestedA$ 
such that infinitely often the following holds:
a slave automaton started at a position $k$ accumulates the value
exceeding $k^2$. 
A run with such a property has an infinite weight according to the semantics
$\flimavg(\pi) = \lim \sup_{k \rightarrow \infty} \frac{1}{k} \cdot \sum_{i=1}^{k} (\cost(\pi))[i]$.
\end{proof}

Now, we are ready to prove decidability of the universality problem for functional 
$(\flimavg;\fsum^+)$-automata.

\begin{proof}[Proof of (2) from Theorem~\ref{thm:crucial-undec}]
If (1) holds, $\nestedA$ is equivalent to a functional $(\flimavg;\fsum^B)$-automaton $\nestedA'$, where $B = {\boundOnWeights \cdot \confNum{\nestedA}}$.
The size of $\nestedA'$ is exponential in $|\nestedA|$. 
The universality problem for functional $(\flimavg;\fsum^B)$-automata is $\PSPACE$-complete, which implies 
the the universality problem for functional  $(\flimavg;\fsum^+)$-automata is in $\EXPSPACE$.
Otherwise, if (2) holds, then an answer to the universality problem for $\nestedA$ is ``No'' for every $\lambda$.
Now, it can be detected whether (1) or (2) holds by reduction to the universality problem
for functional $(\flimsup;\fsum^+)$-automata, which is $\PSPACE$-complete.
\end{proof}

\begin{remark}
The size of $\nestedA'$ is polynomial in the size of the master automaton of $\nestedA$. Therefore, the universaility problem is in $\PSPACE$.
The universality problem for functional $\fsum^+$-weighted automata is $\PSPACE$-hard, hence 
the universality problem for $(\flimavg;\fsum^+)$-automata is $\PSPACE$-complete assuming that the total size of slave automata is bounded.
\end{remark}

\subsection{Summary and Open problems}\label{subsec:summary} 

While we have established the decidability and undecidability of the decision 
problems for nested weighted automata for almost all cases, there is one 
open problem which we present as a conjecture.

\begin{conjecture}
The emptiness problem  for non-deterministic $(\flimavg; \fsum)$-automata is decidable.
\label{conj1}
\end{conjecture}

Tables~\ref{tab1}~and~\ref{tab2} summarize our results.

\noindent{\em Complexity.} The decision problems are 
$\PSPACE$-complete, in $\EXPSPACE$, or undecidable.
We show in Theorem~\ref{thm:succinct} that (deterministic) nested weighted automata are exponentially 
more succinct than (non-deterministic) weighted automata, which explains
$\EXPSPACE$ complexity of some universality problems.

We present the proof of the emptiness case from
Theorem~\ref{thm:limavgsum} in Section~\ref{s:proof-limavg}.

\begin{table}
\centering
\def\tabcolsep{5pt}
\begin{tabular}{|c|c|c|c|c|}
\hline 
\multicolumn{2}{|c|}{}& $\finf$& $\fsup$ & \multirow{2}{*}{$\flimavg$} \\
\multicolumn{2}{|c|}{}& $\fliminf$ & $\flimsup$  & \\  
\hline
$\fmin, \fmax$ &Empt.& 
\multicolumn{3}{|c|}{ \multirow{2}{*}{$\PSPACEshort$-c~\tabref{thm:dec1}}}  \\
\cline{2-2}  
$\fBsum{B}$                               &Univ.&  \multicolumn{3}{|c|}{}  \\
\hline
\multirow{2}{*}{$\fsum$} & Empt. & $\PSPACEshort$-c~\tabref{thm:additional}  & \undecidable~\tabref{thm:crucial-undec} & \multirow{2}{*}{\conjectureOne} \\
\cline{2-4}
& Univ. & \undecidable~\tabref{thm:crucial-undec} & $\PSPACEshort$-c~\tabref{thm:additional}  &  \\
\hline
\multirow{2}{*}{$\fsum^+$} & Empt. & 
\multicolumn{2}{|c|}{ \multirow{2}{*}{$\PSPACEshort$-c \tabref{thm:dec1}}} &
\multirow{2}{*}{$\EXPSPACEshort$~\tabref{thm:limavgsum}} \\
&Univ. & \multicolumn{2}{|c|}{}  &  \\
\hline
\end{tabular}
\caption{Decidability and complexity of the emptiness and universality problems for functional $(f;g)$-automata. 
Functions $f$ are listed in the first row and functions $g$ are in the first column.
The undecidability results hold even for deterministic automata.
Next to each result there is a reference to the corresponding theorem or conjecture.
$\PSPACEshort$ (resp. $\EXPSPACEshort$) denotes $\PSPACE$ (resp. $\EXPSPACE$).}
\vspace{-30pt}
\label{tab1}
\end{table}

\begin{table}
\centering
\def\tabcolsep{5pt}
\begin{tabular}{|c|c|c|c|c|}
\hline 
\multicolumn{2}{|c|}{}& $\finf$& $\fsup$ & \multirow{2}{*}{$\flimavg$} \\
\multicolumn{2}{|c|}{}& $\fliminf$ & $\flimsup$  & \\  
\hline
$\fmin, \fmax$ &Empt.& 
\multicolumn{3}{|c|}{{$\PSPACEshort$-c~\tabref{thm:dec1}}}  \\
\cline{2-5}
$\fBsum{B}$                                &Univ.& \multicolumn{2}{|c|}{{$\EXPSPACEshort$~\tabref{thm:dec1}}} & \undecidable~\tabref{thm:simple} \\
\hline
\multirow{2}{*}{$\fsum$} & Empt. & $\PSPACEshort$-c~\tabref{thm:additional} & \undecidable~\tabref{thm:crucial-undec} &  
\conjectureOne \\
\cline{2-5}
& Univ. & \undecidable~\tabref{thm:crucial-undec} & \undecidable~\tabref{thm:simple} & \undecidable~\tabref{thm:simple} \\
\hline
\multirow{2}{*}{$\fsum^+$} & Empt. & 
\multicolumn{2}{|c|}{$\PSPACEshort$-c~\tabref{thm:additional}}  & $\EXPSPACEshort$~\tabref{thm:limavgsum} \\
\cline{2-5}
&Univ. & \multicolumn{2}{|c|}{$\EXPSPACEshort$~\tabref{thm:additional}}  & \undecidable~\tabref{thm:simple} \\
\hline
\end{tabular}
\caption{Decidability and complexity of the emptiness and universality problems for non-deterministic $(f;g)$-automata. 
$\PSPACEshort$ (resp. $\EXPSPACEshort$) denotes $\PSPACE$ (resp. $\EXPSPACE$).
The alignment is as in Table~\ref{tab1}.}
\vspace{-35pt}
\label{tab2}
\end{table}

\begin{table}
\centering
\def\tabcolsep{5pt}

\begin{tabular}{|c|c|c|c|c|}
\hline 
\multicolumn{2}{|c|}{}& $\finf$& $\fsup$ & \multirow{2}{*}{$\flimavg$} \\
\multicolumn{2}{|c|}{}& $\fliminf$ & $\flimsup$  & \\  
\hline
$\fmin, \fmax$ &Empt.& 
\multicolumn{3}{|c|}{ {$\PTIME$ }}  \\
\cline{2-5}  
$\fBsum{B}$                               &Univ.&  \multicolumn{3}{|c|}{$\PSPACE$-c}  \\
\hline
\multirow{2}{*}{$\fsum$} & Empt. & $\PTIME$  & \undecidable  & \multirow{2}{*}{\conjectureOne} \\
\cline{2-4}
& Univ. & \undecidable   & $\PSPACE$-c  &  \\
\hline
\multirow{2}{*}{$\fsum^+$} & Empt. & 
\multicolumn{3}{|c|}{ {$\PTIME$ }} \\
\cline{2-5}  
&Univ. & \multicolumn{3}{|c|}{$\PSPACE$-c}   \\
\hline
\end{tabular}
\caption{Decidability and complexity of the emptiness and universality problems for functional $(f;g)$-automata whose slave
automata have size bounded by a constant.
The alignment is as in Table~\ref{tab1}.}
\vspace{-35pt}
\label{tab3}
\end{table}

\begin{table}
\def\tabcolsep{5pt}
\centering

\begin{tabular}{|c|c|c|c|c|}
\hline 
\multicolumn{2}{|c|}{}& $\finf$& $\fsup$ & \multirow{2}{*}{$\flimavg$} \\
\multicolumn{2}{|c|}{}& $\fliminf$ & $\flimsup$  & \\  
\hline
$\fmin, \fmax$ &Empt.& 
\multicolumn{3}{|c|}{{$\PTIME$ }}  \\
\cline{2-5}
$\fBsum{B}$                                &Univ.& \multicolumn{2}{|c|}{{$\PSPACE$-c }} & \undecidable \\
\hline
\multirow{2}{*}{$\fsum$} & Empt. & $\PTIME$ & \undecidable &  
\conjectureOne \\
\cline{2-5}
& Univ. & \undecidable  & \undecidable & \undecidable \\
\hline
\multirow{2}{*}{$\fsum^+$} & Empt. & 
\multicolumn{3}{|c|}{$\PTIME$ } \\
\cline{2-5}
&Univ. & \multicolumn{2}{|c|}{$\PSPACE$-c }  & \undecidable \\
\hline
\end{tabular}
\caption{Decidability and complexity of the emptiness and universality problems for non-deterministic $(f;g)$-automata
whose slave automata have size bounded by a constant.
The alignment is as in Table~\ref{tab1}.}
\vspace{-15pt}
\label{tab4}
\end{table}

\noindent{\bf Discussion on inclusion.}
The emptiness and universality problems reduce to the inclusion problem, where
the inclusion problem given two automata $\nestedA_1$ and $\nestedA_2$ asks
whether for every word $w$ we have $\lang_{\nestedA_1}(w) \leq \lang_{\nestedA_2}(w)$.
Therefore, for decidability of the inclusion problem both the emptiness and 
the universality problem must be decidable.
Hence, in the non-deterministic case, for value functions
studied in Table~\ref{tab2}, the inclusion problem can be decidable only in two cases:
\begin{enumerate}
\item for $(f;g)$-automata, where $g$ is regular value function, and 
$f \in \InfVal \setminus \{\flimavg\}$;
\label{cases-inclusion-1st}
\item for $(f;\fsum^+)$-automata, where $f \in \InfVal \setminus \{\flimavg\}$.
\label{cases-inclusion-2nd}
\end{enumerate}

In fact, in case~\eqref{cases-inclusion-2nd}, the inclusion problem is
undecidable as well.
Indeed, inclusion of $\fsum^+$-automata over finite words reduces to the
inclusion of  $(f;\fsum^+)$-automata, where $f \in \{\finf,\fliminf,\fsup, \flimsup\}$.
It has been shown in~\cite{DBLP:conf/atva/AlmagorBK11} that the inclusion
problem for $\fsum^+$-automata is undecidable. Therefore, the inclusion problem in
case~\eqref{cases-inclusion-2nd}  is undecidable.
As automata in case~\eqref{cases-inclusion-1st} are equivalent to
$\silent{f}$-automata for $f \in \{\finf,\fliminf,\fsup,\flimsup\}$ 
(Lemma~\ref{t:regular-non-nested}), which are essentially equivalent to $f$-automata, 
the inclusion problem is decidable~\cite{Chatterjee08quantitativelanguages}.

\begin{remark}[Parametric complexity]
\label{rem:parametric}
The complexity results summarized in Tables~\ref{tab1}~and~\ref{tab2} are given w.r.t. 
the size of a nested automaton, i.e.,
the sum of the size of the master automaton and the total size of slave automata. 
However, if the total size of slave automata is bounded by a constant,
then we show that the complexity of all emptiness problems decreases 
from $\PSPACE$ (resp. $\EXPSPACE$) to $\PTIME$, and all the universality 
problems become $\PSPACE$-complete (Remarks~\ref{rem:firstParRemark}~and~\ref{rem:secondParRemark}). 
In other words, 
we show that the complexity of 
emptiness and universality in the size of the master automaton (with the total size of slave automata considered as constant) matches that of Boolean non-nested automata.
(For every $f\in \InfVal$ the universality problem for functional $f$-automata is 
$\PSPACE$-complete~\cite{DBLP:conf/concur/FiliotGR12}).
Interestingly, bounding the total size of slave automata does not change decidability status; 
all undecidability results still hold.
The parametric complexity results are summarized in Table~\ref{tab3} and Table~\ref{tab4}.
\end{remark}

\FloatBarrier

\section{Applications}
\newcommand{\sr}{\mathsf{sr}}
\newcommand{\ARC}[1]{\mathsf{ARC}(#1)}

In this section we discuss several applications of nested weighted 
automata.

\subsection{Quantitative system properties}
\label{s:quantitativeProperties}
We have shown (Example~\ref{ex:avg-resp-time}) that basic properties such as average response time can be 
expressed conveniently as a nested weighted automaton.
We also argue that our framework is a natural extension of the framework
of monitor automata for Boolean verification, and is a step towards
quantitative run-time verification.

\smallskip\noindent{\em Quantitative monitor automata.}
In verification of Boolean properties, the formalism with \emph{monitor automata}
is a very convenient way to express system properties~\cite{DBLP:conf/spin/PnueliZ08}.
The specification for a system can be decomposed into subproperties, 
each monitor automaton tracks a subproperty, and the logical value of the specification is 
inferred from the results of the monitor automata. 
To be more specific, given an LTL specification, the logical value of every subformula is tracked by a monitor automaton.
A monitor automaton is a transducer that at each position of the word outputs whether the current suffix satisfies the given subformula. 
The monitor automata for complex formulae are constructed from monitor automata for 
their immediate subformulae. Finally, the answer whether a given word satisfies the specification
is encoded as the first output of the monitor that corresponds to the whole LTL formula.
Our nested weighted automata framework can be seen as a natural extension of the formalism
provided by monitor automata.
Below we argue how nested weighted automata provide a convenient framework
for specification, with added expressiveness, and is a first step towards
quantitative run-time verification. 

\begin{itemize}

\item \emph{Ease of specification.}
A specification formalism is a convenient framework if complex specifications
can be easily decomposed. 
For Boolean properties, monitor automata were introduced for this purpose: in other
words, for Boolean properties, though monitor are not more expressive than 
the standard automata, yet they are widely used as they provide a framework 
where specifications can be conveniently described.
In our setting, in the spirit of monitor automata,
each slave automaton can specify a subproperty of the system, and the master automaton
combines the result obtained from all the slave automata.
This (as in the case of monitor automata) allows decomposing quantitative properties into 
subproperties and thus eases the task of specification.
For example as shown in Example~\ref{ex:avg-resp-time} to compute average response time, for each request
the master automaton invokes a slave automaton that computes the response time (a subproperty
for every request) and then the master automaton with limit-average value function 
combines the subproperties to obtain the average response time.

\begin{example}[Average resource consumption]
\label{e:resourceConsumption}
Consider a system with at most $n$ concurrently running processes, 
in which processes can be started and terminated.
The system has available resources $r_1, \ldots, r_k$.
The quantitative property of \emph{average resource consumption}, which asks what is the 
average number of different resources allocated by processes, can be expressed in 
a convenient way by a (deterministic) nested weighted automaton $\nestedA$ defined as follows. 
The master automaton of $\nestedA$ starts a separate slave automaton $\slaveA$ when a new process is started. 
The slave automaton $\slaveA$ runs until the process terminates and counts how many different
resources $r_1, \ldots, r_k$ the given process allocates.
The counting can be implemented by a $\fmax$-automaton with weights $\{0,1,\ldots, k\}$.
Then, the master automaton computes the limit average of resource consumption computed by slave automata.
Since every $(\flimavg; \fmax)$-automaton is equivalent to some
$\silent{\flimavg}$-automaton (Lemma~\ref{t:regular-non-nested}),
average resource consumption can also be expressed by a weighted automaton. 
However, construction of such a weighted automaton is cumbersome and 
it essentially follows the proof of Lemma~\ref{t:regular-non-nested}.
\end{example}

We use Example~\ref{e:resourceConsumption} to show that (deterministic) nested 
weighted automata can be exponentially more succinct than (non-deterministic) weighted automata.
Let $\ARC{n}$ denote the average resource consumption property for $n$-processes.

\begin{restatable}{theorem}{NestedAreSuccinct}
\label{thm:succinct}
There is a deterministic $(\flimavg; \fmax)$-automaton of the size $O(n)$ expressing $\ARC{n}$,
while every non-deterministic $\silent{\flimavg}$-automaton expressing $\ARC{n}$ has $2^{\Omega(n)}$ states.
\end{restatable}
\begin{proof}
\newcommand{\AllocA}{\Simga_{\textrm{alc}}}
Recall the the automaton $\nestedA$ from Example~\ref{e:resourceConsumption} that expresses
average resource consumption. 
Its size is linearly bounded in the number of processes $n$. 
It remains to show that  
every $\silent{\flimavg}$-automaton expressing $\ARC{n}$,
average resource consumption for $n$ processes, has $2^{\Omega(n)}$ states.
To show that, we need to give a more precise description of the system in consideration 
and its modeling. 

We assume for simplicity that there is only a single resource.
Each process $i \in \{1,\ldots, n\}$ is associated with the following actions:
\emph{start} ($s_i$),
\emph{allocation} of the resource ($a_{i}$),  and \emph{termination} ($t_i$).
Formally, average resource consumption in $w$ is defined as the limit average 
over all positions $p$ at which a process starts $w[p] = s_i$ (for some $i$)
of the indicator ($0$/$1$) whether $a_i$ occurs in $w$ between position $p$ and 
the first occurrence of $t_i$ past $p$.

We show that unless a $\silent{\flimavg}$-automaton has at least $2^{0.5 n}$ states, it cannot 
compute average resource consumption.
Assume towards contradiction that a $\silent{\flimavg}$-automaton $\nonnestedA$
has less than $2^{0.5n}$ states and computes average resource consumption.
For every $A \subseteq \{ a_1, \ldots, a_n\}$, we define a word $u_A \in A^*$
as a periodic listing of all letters from $A$ $|\nonnestedA|$ times, i.e.,
$(b_1 \ldots b_s)^{|\nonnestedA|}$, where $\{ b_1, \ldots, b_s \} = A$.
Consider execution traces $w_A = (s u_A t)^{\omega}$, where $s = s_1 \ldots s_n$,
$t = t_1 \ldots t_n$.
Given a word $w_A$, let $\pi_A$ be a run of $\nonnestedA$ on $w_A$ of the minimal value.
Due to periodicity of $w_u$, such a run exists.
We show the following claim: 
\smallskip

\noindent(*)~
There exist cycles $c_A, c_B$ in the automaton $\nonnestedA$ such
that (1)~ $c_A, c_B$ are labeled with words over different alphabets $A, B$, 
with $|A| = |B| = 0.5n$,
(2)~ $c_A, c_B$ share a state $q$ that occurs in both runs $\pi_A, \pi_B$ with positive density, and
(3)~each of $c_A, c_B$ is either silent or its average weight is $0.5n$.
\smallskip

We shall prove (*) later; first we show that (*) implies that
$\nonnestedA$ does not express average resource consumption.
Indeed, consider a pair $c_A, c_B$ from (*) and $a_i$ from $B \setminus A$.
We insert into the run $\pi_A$ the cycle $c_B$
at all positions where $q$ occurs.
  Let $\pi_A'$ be the resulting run, and let $w_A'$ be the word
that corresponds to $\pi_A'$.
The resulting run $\pi_A'$ has the same value
as $\pi_A$, $\frac{1}{2}$, but average resource consumption in $w_A'$ is higher.
Indeed, in all blocks $s u t$, for some $u$, in which $u$ is different from $u_A$,
the number of different letters in $u$ is at least $|A|+1$.
In the remaining blocks, the number of different letters is $|A|$.
Since blocks $s u t$ with $u \neq u_A$ occur with positive density, average resource 
consumption in $w_A'$ is strictly  higher than $\frac{1}{2}$.
But, the value of $w_A'$ does not exceed  $\frac{1}{2}$ as
$\pi_A'$ is an accepting run on $w_A'$ of the value $\frac{1}{2}$.
It follows that $\nonnestedA$ does not express the average resource consumption property.

Now, we prove (*). Consider a word $w_A$ and
an occurrence of $u_A$ at position $p$ in $w_A$.
Since $|u_A| > |\nestedA| + 2 \cdot |A|$, there is a  state $q$
 that occurs twice in the run $\pi_A$ between positions $p$ and $p + |u_A|$ 
(the part corresponding to the considered occurrence of $u_A$) and
the distance between occurrences of $q$ between $|A|$ and $|u_A| - |A|$. 
These occurrences of $q$ indicate a cycle, $c_{A,p}$,  in $\nestedA$,
which is labeled with  all letters from $A$. 
Indeed, among any $|A|$ consecutive letters in $u_A$ each letter from $A$ occurs.
Now, we select $c_A$ from cycles $c_{A,i}$, where $i$ varies, that occurs with positive density in $\pi_A$. 
Let $q_A$ be the state that occurs in $c_{A}$.

As there are more than $2^{0.5n}$ subsets of $\{a_1, \ldots, a_n\}$ of cardinality $0.5n$, 
there is a state that occurs in two cycles
$c_A, c_B$ with $A \neq B$. It remains to show that $c_A, c_B$
are either silent or their average weight is $\frac{1}{2}$. 
If the average weight of 
$c_A$ is greater than $\frac{1}{2}$, we can decrease the value of $\pi_A$
by removing all occurrences of $c_A$.
Recall that the length of $c_A$ is at most $|u_A| - |A|$, therefore 
if we remove all parts of $w_A$ that correspond to $c_A$, each process in the
resulting word has resource consumption $0.5n$, hence average resource consumption is still $\frac{1}{2}$.
But, the value of the corresponding run is lower than $\frac{1}{2}$, a contradiction.
Conversely, if the value is lower than $\frac{1}{2}$, we can pump that cycle to
obtain a run of the value smaller than $\frac{1}{2}$ on a word whose resource consumption
for each process is $0.5n$.
Thus, $c_A$ is either silent or its average weight value is $\frac{1}{2}$.
\end{proof}

\item \emph{Expressiveness.} More importantly, as mentioned above, for Boolean properties,
monitor automata only add convenience but not expressiveness, whereas we show that 
for quantitative properties, nested weighted automata are strictly more expressive than 
non-nested weighted automata.
Moreover, we show that the added expressiveness of nested weighted automata comes with 
the ability to express natural quantitative properties (like average response time) 
that could not be expressed as non-nested weighted automata.

\item \emph{Quantitative run-time verification.}
Finally, monitor automata are specially useful for safety properties, and widely used in 
run-time verification~\cite{DBLP:conf/tacas/HavelundR02}.
Our nested weighted automata can be seen as the first step towards quantitative
run-time verification.
Each slave automaton acts as a monitor and returns values of subproperties of the system.
If the value function of the master automaton is commutative (as in all our examples),
the master automaton can compute an on-the-fly approximation of the value function 
for finite words.
\end{itemize}

\subsection{Model measuring}
\label{s:applications-model-measuring}

The \emph{model-measuring} problem~\cite{myconcur} asks, given a model and a 
specification, what is the maximal distance $\rho$ such that 
all models within distance $\rho$ from the model satisfy the specification.
Formally, a model $M$ and a specification $S$ are Boolean automata.
Given $M$,  a \emph{similarity measure} (of $M$) is a function $d_M$ from infinite 
words to positive real numbers such that for all traces $w$ in $\lang_M$ we have 
$d_{M}(w) = 0$. 
Similarity measures extend to models in a natural way; i.e., 
$d_{M}(M') = \sup \{ d_M(w) : w$ is a trace of $M' \}$.
The \emph{stability radius} of $S$ in $M$ w.r.t. the similarity measure $d_M$, 
denoted by $\sr_{d_M}(M, S)$, is defined as 
$\sr_{d_M}(M, S) = \sup \{ \rho \geq 0 : \forall M' (d_M(M') <~\rho \Rightarrow \lang_M \subseteq \lang_S) \}$.
We are interested in similarity measures $d_M$ defined by nested weighted automata (resp. weighted automata as in~\cite{myconcur}).
Note that $d_M$ is independent of the specification.
The model-measuring decision problem of whether $\sr_{d_M}(M,S) \leq \lambda$
reduces to the emptiness decision question~\cite{myconcur}.
We now show how nested weighted automata can define interesting 
similarity measures $d_M$.

\begin{example}[Bounded delays]
Consider the model $M$ for two processes communicating through a channel, 
where every sent packet is delivered in the next state.
Let $\$$ denote the event of neither sending or receiving packets, 
$s_1$ and $r_1$ (resp. $s_2$ and $r_2$) the send and receive for process~1 
(resp. process~2).
The language of $M$ can be described as a regular expression as follows: 
$( (\$)^* \cdot (s_1 r_1)^* \cdot (\$)^* \cdot (s_2 r_2)^*)^{\omega}$.

Note that $d_M$ must assign value~0 to every trace in the language of $M$.
Also $d_M$ needs to assign values to traces where the delivery of packets can 
be delayed by a finite amount. 
Hence we first need to relax the language of $M$ as $M_R$ such that every packet 
sent is received with a finite delay; and $d_M$ assigns values to traces in the 
language of $M_R$.
The relaxed language $M_R$ is obtained as follows:
consider the following languages $L_1$ and $L_2$ 
\[
L_1 = ( \$^* \cdot (s_1 \$^* r_1)^* \cdot \$^*)^{\omega}; 
 \text{ and } 
L_2 = ( \$^* \cdot (s_2 \$^* r_2)^* \cdot \$^*)^{\omega};
\]
where $L_1$ denotes that every sent for process~1 can be delayed by a 
finite amount and analogously $L_2$ for process~2.
The language of $M_R$ is the \emph{shuffle} (arbitrary interleavings) 
of $L_1$ and $L_2$.

The similarity measure $d_M$ is defined as a $(\fsup;\fsum^+)$-automaton $\nestedA_D$ 
that computes the maximum delay in the following way. 
When a packet is sent, the master automaton starts a slave
$\fsum^+$-automaton that counts the number of transition until the packet is delivered. 
If no packet is sent, the master automaton takes a silent transition.
The product automaton ${M_R}$ and $\nestedA_D$ defines the desired similarity 
measure.
\end{example}

\subsection{Model repair}
\label{s:applications-inner-model-measuring}
The \emph{model-repair} problem, given a model and a specification, 
asks for the minimal restriction of the model such that the specification 
is satisfied.
Given a model $M$, a \emph{repair measure} $d_M$ is a function 
from infinite words to real numbers such that 
$d_M(w) < \infty$ iff $w\in \lang_M$. 
Intuitively, the measure evaluates the hardness of traces of $M$, 
which can be used to evaluate severity of the violation of the specification.
We are interested in $d_M$ specified by nested weighted automata (resp. weighted automata).
Given a model $M$, a repair measure $d_M$, and a real number $r$, 
we define the language $d_M^{<r}$ as $\{ w : d_M(w) < r \}$.
The model-repair decision problem, given a model $M$, 
a repair measure $d_M$, and a specification $S$, asks whether 
$\sup \{r : d_M^{<r} \subseteq \lang_S \} \leq \lambda$.
The model-repair decision problem also reduces to the emptiness question.

\begin{example}[Context-switches]
Consider a system consisting of a scheduler and two programs.
The scheduler starts processes infinitely often and does
preemptive scheduling. 
To obtain a finite-state model, we consider that only a single instance of 
each program may run at a given time.
Consider the repair measure $d_M$ that represents the negative of the \emph{minimal slot length}, i.e., 
for all $w$ we have $d_M(w)=-k$ iff each process in the execution $w$ runs for at least $k$ steps.  
The repair measure can be defined by a functional $(\fsup; \fsum)$-automaton $\nestedA_R$
as follows. After each context-switch, the master automaton starts an automaton that computes 
the running time until the next context-switch and multiplies it by $-1$ (i.e., add~$-1$ at each 
step). 
At steps at which there is no context switch, the master automaton takes a silent transition.
It follows that the supremum of all those values is the length of the shortest running time
of a process multiplied by $-1$.
Although, the emptiness problem is undecidable for $(\fsup; \fsum)$-automata,
the automaton $\nestedA_R$ has only non-positive weights. 
The emptiness problem for $(\fsup; \fsum)$-automata with non-positive weights reduces to the universality 
problem for $(\finf; \fsum^+)$-automata, which is decidable.
\end{example}

\begin{remark}[Decidability of examples]
Note that for all examples presented in the paper, they belong 
to the class of nested weighted automata for which we establish decidability of
the emptiness problem.
\end{remark}

\begin{remark}[Robustness of nested weighted automata]
The model of nested weighted automata is robust with respect
to several changes, e.g., (i)~instead of labeling function on transitions 
we can have labeling function on states; or 
(ii)~instead of invoking one slave automaton in every transition a constant 
number of slave automata can be invoked. 
These changes do not change the expressive power, nor the 
decidability and the complexity results for nested weighted automata.
\end{remark}

\section{Emptiness of $(\flimavg;\fsum^+)$-automata is in $\EXPSPACE$}
\label{s:proof-limavg}

\newcommand{\BoNumberOfSA}{\mathbf{c}}
\newcommand{\BoValueOfSA}{\mathbf{d}}
\newcommand{\conf}{{\cal C}}
\newcommand{\mult}{{\mathsf{mult}}}

\newcommand{\boundOnMulti}{\mathbf{N}}
\newcommand{\dom}{\mathrm{dom}}

\newcommand{\tikzAccummulation}{
\node[rectangle,draw] (w1) at (0,0) {$2$};
\node[rectangle,draw] (w2) at (0.5,0) {$3$};
\node[rectangle,draw] (w3) at (1,0) {$1$};
\node[rectangle,draw] (w4) at (1.5,0) {$0$};
\node[rectangle,draw] (w5) at (2,0) {$0$};
\node at (3.7,0) {master automaton}; 
\node[rectangle,draw] (s1) at (0,2.2) {$0$};
\node[rectangle,draw] at (0.5,2.2) {$1$};
\node[rectangle,draw] at (1,2.2) {$0$};
\node[rectangle,draw] at (1.5,2.2) {$1$};
\node[rectangle,draw, thick,minimum height=0.51cm, minimum width=2.1cm]  at (0.7,2.2) {};
\node at (3.7,2.2) {slave automaton 1}; 
\node[rectangle,draw] (s2) at (0.5,1.65) {$1$};
\node[rectangle,draw] at (1,1.65) {$0$};
\node[rectangle,draw] at (1.5,1.65) {$1$};
\node[rectangle,draw] at (2,1.65) {$1$};
\node[rectangle,draw, thick,minimum height=0.51cm, minimum width=2.1cm] at (1.2,1.65) {};
\node at (3.7,1.65) {slave automaton 2}; 
\node[rectangle,draw] (s3) at (1,1.1) {$0$};
\node[rectangle,draw] at (1.5,1.1) {$1$};
\node[rectangle,draw] at (2,1.1) {$0$};
\node[rectangle,draw, thick, minimum height=0.51cm, minimum width=1.6cm] at (1.5,1.1) {};
\node at (3.7,1.1) {slave automaton 3}; 
\draw[->] (s1) to (w1);
\draw[->] (s2) to (w2);
\draw[->] (s3) to (w3);
\begin{scope}[xshift=5.8cm]
\node[rectangle,draw] (w1) at (0,0) {$0$};
\node[rectangle,draw] (w2) at (0.5,0) {$2$};
\node[rectangle,draw] (w3) at (1,0) {$0$};
\node[rectangle,draw] (w4) at (1.5,0) {$3$};
\node[rectangle,draw] (w5) at (2,0) {$1$};
\node at (0.9,-0.5) {simulation weights}; 
\node[rectangle,draw] (s1) at (0,2.2) {$0$};
\node[rectangle,draw] (s2) at (0.5,2.2) {$1$};
\node[rectangle,draw] (s3) at (1,2.2) {$0$};
\node[rectangle,draw] (s4) at (1.5,2.2) {$1$};
\node[rectangle,draw] (s2) at (0.5,1.65) {$1$};
\node[rectangle,draw] at (1,1.65) {$0$};
\node[rectangle,draw] at (1.5,1.65) {$1$};
\node[rectangle,draw] at (2,1.65) {$1$};
\node[rectangle,draw] (s3) at (1,1.1) {$0$};
\node[rectangle,draw] at (1.5,1.1) {$1$};
\node[rectangle,draw] at (2,1.1) {$0$};
\node[rectangle,draw, thick, minimum height=1.2cm, minimum width=0.51cm] (s2) at (0.5,1.9) {};
\node[rectangle,draw, thick, minimum height=1.7cm, minimum width=0.51cm] (s3) at (1.0,1.65) {};
\node[rectangle,draw, thick, minimum height=1.7cm, minimum width=0.51cm] (s4) at (1.5,1.65) {};
\node[rectangle,draw, thick, minimum height=1.2cm, minimum width=0.51cm] (s5) at (2,1.4) {};
\draw[->] (s2) to (w2);
\draw[->] (s3) to (w3);
\draw[->] (s4) to (w4);
\draw[->] (s5) to (w5);
\end{scope}
}

\newcommand{\cond}{{\mathfrak C}}
\newcommand{\nestedAMod}{\nestedA^{\#}}
\newcommand{\nonnestedAMod}{\nonnestedA^{\#}}
\newcommand{\Qslv}{Q_{\textrm{slv}}}

In this section we prove that the emptiness problem for 
$(\flimavg;\fsum^+)$-automata is in $\EXPSPACE$ ((1) from Theorem~\ref{thm:limavgsum}), as the proof 
itself is interesting and requires new and non-standard techniques.
We first present an overview of the proof.

\smallskip\noindent{\em Overview of the proof.}
The key argument will be to \emph{simulate} a given $(\flimavg;\fsum^+)$-automaton $\nestedA$ 
by a $\silent{\flimavg}$-automaton, however, the main conceptual difficulty is that
$(\flimavg;\fsum^+)$-automata are strictly more expressive than $\silent{\flimavg}$-automata.
We circumvent this problem (which is non-standard for weighted automata) in the following way:
\begin{enumerate}
\item \emph{Step~1.} 
We establish a property $\cond$ on runs of a $(\flimavg;\fsum^+)$-automaton $\nestedA$ such that 
(a)~the infimum over values of runs satisfying $\cond$ is the same as the infimum over values of all runs, and 
(b)~there is a $\silent{\flimavg}$-automaton that simulates $\nestedA$ on runs satisfying $\cond$.

\item \emph{Step~2.} We give the construction of a $\silent{\flimavg}$-automaton $\nonnestedA$
specified in the condition (b) from Step~1.

Although, $\nonnestedA$ simulates $\nestedA$, weighted automata and nested weighted automata 
accumulate weights in a different way;
a run of $\nestedA$ that satisfies $\cond$ and
the corresponding run of $\nonnestedA$ can have different values. 

\item \emph{Step~3.} We show that the infima over values of all runs 
of $\nestedA$ and $\nonnestedA$ are equal. 
\end{enumerate}

\smallskip\noindent{\bf Proof of Step~1.} We first introduce the notion of bounded multiplicity.

\smallskip\noindent{\em Configuration and multiplicities.}
In nested weighted automata, starting a slave automaton can be seen as a universal transition 
in the sense of alternating automata. 
We adapt the power-set construction, which is used to convert alternating
automata to non-deterministic automata, to the nested weighted automata case. 
Given a nested weighted automaton $\nestedA$, we define \emph{configurations} and \emph{multiplicities} of $\nestedA$
as follows.
Let $\Qslv$ be the disjoint union of the sets of states of all slave automata of $\nestedA$.
For a run of $\nestedA$, we say that $(q_m, A)$ is the \emph{configuration} at position $p$ if
$q_m$ is the state of the master automaton at position $p$
and $A \subseteq \Qslv$ is the set of states of slave automata at position $p$.
We denote by $\confNum{\nestedA}$ the number of configurations of $\nestedA$.
We define the \emph{multiplicity} $\mult$ at position $p$ as the function
$\mult : \Qslv \mapsto \N$, such that $\mult(q)$ specifies the number of 
slave automata in the state $q$ at position $p$.
The configuration together with the multiplicity give a complete 
description of the state of $\nestedA$ at position $p$.

\smallskip\noindent{\em Optimal runs.}
The general idea to solve the emptiness problem for $(\flimavg;\fsum^+)$-automata is
to simulate a given $(\flimavg;\fsum^+)$-automaton  by a $\silent{\flimavg}$-automaton that keeps track of 
configurations and multiplicities. 
Unfortunately, unbounded multiplicities cannot be encoded in a finite set of states 
of a $\silent{\flimavg}$-automaton. 
But, the emptiness problem can be solved by inspecting only selected runs.
More precisely, given a $(\flimavg;\fsum^+)$-automaton $\nestedA$, we say that runs
satisfying a condition $\cond$ are \emph{optimal} for $\nestedA$ iff 
the infima over values of all runs of $\nestedA$ and runs that satisfy $\cond$
are equal. We identify a condition such that the runs satisfying it are optimal and
can be simulated by a $\silent{\flimavg}$-automaton.
First, we observe that 
without loss of generality we can assume that nested weighted automata are deterministic. 
Basically, non-deterministic choices can be encoded in the input alphabet.

\begin{restatable}{lemma}{WLOGDeterministicLimAvgSum}
\label{WLOGDeterministicLimAvgSum}
Given a $(\flimavg; \fsum^+)$-automaton $\nestedA$ over $\Sigma$,
one can compute in polynomial space a deterministic $(\flimavg; \fsum^+)$-automaton $\nestedA'$
over an alphabet $\Sigma \times \Gamma$ such that
$\inf_{w \in \Sigma^+} \valueL{\aut}(w) = \inf_{w' \in (\Sigma \times \Gamma)^+} \valueL{\aut'}(w')$.
Moreover, $\confNum{\nestedA} = \confNum{\nestedA'}$.
\end{restatable}
\begin{proof}
The proof consists of two steps. 
We show that (i)~for every run $(\masterRun, \slaveRun_1, \slaveRun_2, \dots)$ of $\nestedA$ there exists
a \emph{simple} run of $\nestedA$ of the value not exceeding the value of $(\masterRun, \slaveRun_1, \slaveRun_2, \dots)$. 
Next, we show that (ii)~there exists a deterministic $(\flimavg; \fsum^+)$-automaton $\nestedA'$ over
an extended alphabet such that 
the sets of accepting simple runs of $\nestedA$ and accepting runs of $\nestedA'$
coincide and each run has the same value in both automata. Then (i) and (ii) imply the lemma statement.

(i): A run of a nested weighted automaton is \emph{simple} if at every position in the run 
slave automata that are in the same state take the same transition.
Now, consider a run $(\masterRun, \slaveRun_1, \slaveRun_2, \dots)$ of $\nestedA$.
Suppose that $\pi_i, \pi_{j}$ that are in the same state 
at the position $s$ in the word, i.e., $\pi_i[i'] = \pi_j[j']$,  where
$i', j'$ are the position in $\pi_i, \pi_j$ corresponding to the position $s$ in $w$.

We choose from the suffixes $\pi_i[i', |\pi_i|], \pi_j[j', |\pi_j|]$ the one with the smaller value
and change the suffixes of both runs to the chosen one. 
If these suffixes have the same value, we chose the shorter one. 
Such a transformation does not increase the value of the partial sums and
does not introduce infinite runs of slave automata. 
Indeed, a run of each slave automaton can be changed by such an operation only finitely many times.
Thus, this transformation can be applied to any pair of slave runs to 
obtain a simple run of the value not exceeding the value of $(\masterRun, \slaveRun_1, \slaveRun_2, \dots)$.

\newcommand{\Qall}{Q_{\textrm{all}}}
(ii): 
Without loss of generality, we can assume that  for every slave automaton in $\nestedA$
final states have no outgoing transitions. 
Let $\Qall$ be the disjoint union of 
the sets of states of the master automaton and all slave automata of $\nestedA$.
We define $\Gamma$ as the set of all partial functions $h : \Qall \mapsto \Qall$.
We define a $(\flimavg; \fsum^+)$-automaton $\nestedA'$ over the alphabet $\Sigma \times \Gamma$ 
by modifying only the transition relations and labeling functions of the master automaton and slave automata of $\nestedA$; 
the sets of states and accepting states are the same as in the original automata.
The transition relation and the labeling function of the master automaton $\masterA'$ of
$\nestedA'$ is defined as follows: for all states $q,q'$, 
$(q, \lpair{a}{h},q')$ iff $h(q) = q'$ and $\masterA$ has the transition $(q,a,q')$.
The label of the transition $(q, \lpair{a}{h},q')$ is the same as the label 
of the transition $(q,a,q')$ in $\masterA$.  Similarly, 
for each slave automaton $\slaveA_i$ in $\nestedA$, the transition relation of the corresponding
slave automaton $\slaveA_i'$ in $\nestedA'$ is defined as follows:
for all states $q,q'$ of $\slaveA_i'$, 
$(q, \lpair{a}{h},q')$ iff $h(q) = q'$ and $\slaveA_i$ has the transition $(q,a,q')$.
The label of the transition $(q, \lpair{a}{h},q')$ is the same as the label 
of the transition $(q,a,q')$ in $\slaveA_i$. 

First, we see that $\confNum{\nestedA} = \confNum{\nestedA'}$.
Second, observe that the master automaton $\masterA'$ and all slave automata $\slaveA_i'$
are deterministic. Moreover, since we assumed that for every slave automaton in $\nestedA$
final states have no outgoing transitions, slave automata $\slaveA_i'$ recognize 
prefix free languages. 
Finally, it follows from the construction that
(i)~for every simple run $(\masterRun, \slaveRun_1, \slaveRun_2, \dots)$ of $\nestedA$
is also a run of $\nestedA$ of the same value. One needs to encode non-deterministic transitions in functions $h \in \Gamma$.
The value of each transition is the same by the construction.
Conversely, (ii) a run $(\masterRun, \slaveRun_1, \slaveRun_2, \dots)$ of $\nestedA'$
is a simple run of $\nestedA$ of the same value. Indeed, the fact that transitions are
directed by functions $h \in \Gamma$ implies that the run is simple.
\end{proof}

We attempt to simulate $(\flimavg;\fsum^+)$-automaton $\nestedA$ by a $\silent{\flimavg}$-automaton.
For that, we need to show that runs with bounded multiplicities or bounded values returned by slave automata
are optimal for $\nestedA$; otherwise the state space of the simulating automaton 
would have to be unbounded.
However, such a direct statement does not hold as we see in the following example.

\begin{example}
\label{e:OptimalUnbounded2}
Consider a $(\flimavg;\fsum^+)$-automaton $\nestedA$ over $\{a,b\}$ such that
on letter $a$  (resp. $b$) the master automaton starts a slave automaton $\slaveA_a$ (resp. $\slaveA_b$).
The automaton $\slaveA_a$ accepts words $a^*b$, and for a word $a^k b$ assigns value $k \mod 2$.
The automaton $\slaveA_b$ accepts words $ba^*b$, and for a word $b a^k b$ assigns value $k + 2$. 
Observe that the infimum over all runs of $\nestedA$ is $\frac{3}{2}$.

At the end of a block of $k$ letters $a$, 
the number of slave automata running concurrently is $k+1$, one automaton $\slaveA_b$ and 
$k$ automata $\slaveA_a$, and the  value returned by $\slaveA_b$ is $k+2$. 
It follows that if the multiplicity of a run is bounded by $k+1$
or the maximal returned values are bounded by $k+2$,
lengths of all block of $a$'s are bounded by $k$. 
However, if the length of block's of letter $a$ are bounded by $k$, 
the value of such a run is at least $\frac{3\cdot k + 4}{2 \cdot (k+1)}$. 
Thus, runs of bounded multiplicity 
or bounded returned value are not optimal for $\nestedA$. 
\end{example}

Example~\ref{e:OptimalUnbounded2} shows that we cannot bound the number of 
slave automata running concurrently or the values returned by slave automata. 
However, we can combine these two conditions, i.e., 
we show that while computing the infimum over values of runs of an  $(\flimavg;\fsum^+)$-automaton,
there is a constant $\boundOnMulti$ such that
we can discard runs in which more than $\boundOnMulti$ slave automata accumulate value above $\boundOnMulti$.
Then, slave automata that return bounded values are essentially bounded sum automata,
and can be eliminated, and only bounded number of slave automata returning unbounded values remain.
E.g. the automaton $\nestedA$ from Example~\ref{e:OptimalUnbounded2} is equivalent to a $(\flimavg;\fsum^+)$-automaton
$\nestedAMod$ that, instead of starting a slave automaton $\slaveA_a$, 
guesses the parity of the following block of letters $a$ and, 
based on that guess, starts $\slaveA_{a0}$ or $\slaveA_{a1}$,
which terminates after a single transition and returns $0$ ($\slaveA_{a0}$) or $1$ ($\slaveA_{a1}$).
The master automaton verifies the correctness of the guessed parity.
  
\smallskip\noindent{\em Synchronized silent transitions.} 
The bounded multiplicity property enables simulation of $\nestedA$ by $\silent{\flimavg}$-automaton, but 
it does not guarantee that the corresponding runs have the same value. 
We impose the following condition on optimal runs of $\nestedA$.
We say that a run of $\nestedA$ has \emph{synchronized silent transitions} if at every
position where the master automaton takes a silent transition, each 
slave automaton takes a transition of weight $0$.
On runs of $\nestedA$ with synchronized silent transitions,
the $\silent{\flimavg}$-automaton can take a silent transition
whenever the master automaton of $\nestedA$ takes a silent transition 
as no weighted transition is lost. 
To achieve synchronization of silent transitions,
we modify slave automata so that
during silent transitions of the master automaton, 
slave automata accumulate their values in their states, while taking transitions of weight $0$, 
and flush the accumulated value $\const$ by taking transition of weight $\const$
once the master automaton takes a non-silent transition. 
We prove that runs with sequences of silent transitions 
bounded by $\confNum{\nestedA}$ are optimal for $\nestedA$, therefore
slave automata have to accumulate only bounded weights.

We combine the ideas for bounding the multiplicity and synchronization of silent 
transitions in the following lemma.

\begin{restatable}{lemma}{BoundedWidth}
Let $\nestedA$ be a deterministic $(\flimavg;\fsum^+)$-automaton. 
There is a constant $\BoNumberOfSA$ quadratically bounded in $\confNum{\nestedA}$ and 
a deterministic $(\flimavg;\fsum^+)$-automaton $\nestedA_0$
equivalent to $\nestedA$ such that runs that have
(1)~multiplicities bounded by $\BoNumberOfSA$, and
(2)~synchronized silent transitions, 
are optimal for $\nestedA_0$.
The size of $\nestedA_0$ is exponentially bounded in $|\nestedA|$.
\label{l:bound-on-slave}
\end{restatable}

Before we prove Lemma~\ref{l:bound-on-slave}, we show its vital components:
\begin{lemma}
Let $\nestedA$ be a deterministic $(\flimavg; \fsum^+)$-automaton that
has an accepting run.
Runs such that among every consecutive $\confNum{\nestedA}$ 
steps, the master automaton of $\nestedA$ takes a non-silent transition
are optimal for $\nestedA$.
\label{many-non-silent}
\end{lemma}
\begin{proof}
Consider a run of $\nestedA$ on a word $w$ and positions $i, j$ such that 
$i + 2 \cdot \confNum{\nestedA} < j$ and $\nestedA$ takes only silent transitions between $i$ and $j$.

Observe that there are positions $i < i', j' < j$
with the same configuration (defined in Section~\ref{s:proof-limavg}). 
Consider a word $w'$ resulting from removing $w[i', j']$ from $w$.
The partial sum of the weights of the master automaton up to the position $j-(j'-i')$ on $w'$  does not exceed 
the partial sum up to the position $j$ on $w$. 
These partial sums are divided, in the average, by the same number of steps. 
Thus, the value of the word
will not increase even if we can carry out this operation infinitely often.
 One should be careful not to remove all positions with accepting states. 
However, it is not a serious problem as
we can insert sparsely subwords with an accepting state (after $1,2, \dots, 2^k, \dots$ time increase steps).
Such an operation will not increase the limit average of the run.
\end{proof}

\begin{lemma}
Let $\nestedA$ be a deterministic $(\flimavg;\fsum^+)$-automaton that recognizes a non-empty language.
Let $\boundOnMulti  = (|Q_s|+2) \cdot \confNum{\nestedA}$.
The runs that eventually (for every position $s$ greater than some position $s_0$) satisfy the following condition (*) are are optimal for $\nestedA$:
(*)~among slave automata active at position $s$, at most $2 \cdot \boundOnMulti$ will accumulate value greater than $4 \cdot \boundOnMulti$.
\label{lim-avg-technical}
\end{lemma}
\begin{proof}
\newcommand{\valF}{\mathsf{val}}
For a multiplicity $\mult$ we define its restriction to ${\boundOnMulti}$, $\mult\restriction_{\boundOnMulti}$, as
$\mult\restriction_{\boundOnMulti}(q) = \min(\mult(q), {\boundOnMulti})$, for every $q \in \mathrm{dom}(\mult)$.

Consider any word $uw$ such that
at the position $|u|$ there are $2 \cdot \boundOnMulti$ slave automata 
that will accumulate in $w$ (past the position $|u|$ in $uw$) value greater than $4 \cdot \boundOnMulti$.
We show a transformation of $uw$ to $uw'$, such that $uw'$ has the same value 
and at the position $|u|$ no slave automaton will accumulate in $w'$ value greater than $4 \cdot \boundOnMulti$.
Let $j_0>|u|$ be a position in $uw$ with an accepting state and
let $j_1, \ldots, j_n$ be the positions at which each of
slave automata started before the position $|u|$ finishes.
Note that $n \leq |Q_s|$.
As slave automata work on finite words such $j_1, \ldots, j_n$ exist.
Finally, let $j$ be the first position greater than 
$max(j_0, j_1, \ldots, j_n)$ with the configuration $\conf_{uw[1,j-|u|]} =  \conf_u$ and 
multiplicity $\mult_{uw[1,-|u|]}\restriction_{\boundOnMulti} = \mult_{u}\restriction_{\boundOnMulti}$.
There are only finitely many positions $|u|$ for which such $j$ does not exist.
Next, as there is no bound on $j$, we remove from $w[1,j-|u|]$ all cycles that do not overlap with any position from $\{j_0, \ldots, j_n\}$.
The resulting word $v$ has the length bounded by $(|Q_s|+2) \cdot \confNum{\nestedA} = {\boundOnMulti}$ as $n \leq |Q_s|$.
It follows that for every $q \in A$ we have $\mult_{uv}(q) \leq {\boundOnMulti}$ and
$\mult_{uv}(q) \leq \mult_{u}(q)\restriction_{\boundOnMulti}$.
Indeed, since cycle removal does not increase the multiplicity, 
for every $q \in A$ we have $\mult_{uv}(q) \leq \mult_{uw[1, j]}(q)$ 
and $\mult_{uw[1,j]}\restriction_{\boundOnMulti} = \mult_{u}\restriction_{\boundOnMulti}$.
We show that the partial sum of weights of the master automaton at the position
$|uv|$ in $uvw$ is smaller than the partial sum at the position $|u|$ is $uw$, which implies
that the transformation $uw \rightarrow uvw$, removes a position violating our assumption,
and even applied infinitely many times, does not increase the value of the resulting words.

Let $\valF$ be a function with $\dom(\valF) = \dom(\mult_u)$ such that
$\valF(q)$ is the value accumulated in $w$ (past the position $|u|$ i $uw$) by
any slave automaton that is in the state $q$ at the position $|u|$.
Equivalently, that is the value accumulated by the same automaton past the position $|uv|$ in $uvw$.
We call a slave automaton \emph{active} if $\valF(q) \geq 4 \cdot \boundOnMulti$, where
$q$ is the state of that automaton at the position $|u|$ (resp. $|uv|$).
The value of the partial sum  up to the position $|u|$ in $uw$ is the value of all
slave automata started before $|u|$. It  consists
of $(1) + (2)$, where 
\begin{itemize}
\item $(1)$ is the value all inactive slave automata
plus the value of active slave automata accumulated up to the position $|u|$, and
\item $(2)$ is the value accumulated in $w$ by all active automata past
the position $|u|$. 
\end{itemize}

\noindent Observe that 
$(2) =  \sum_{q \in A} \valF(q) \cdot \mult_{u}(q)$, where $A$ 
is the set of states of active slave automata at the position $|u|$.
The value of the partial sum  up to the position $|uv|$ in $uvw$ consists
of $(1)' + (2)' + (3)$, where 
\begin{itemize}
\item $(1)'$ is the value of all inactive slave automata
plus the value of active slave automata accumulated up to the position $|u|$,
\item  $(2)'$ is the value accumulated by active automata in $w$
past position $|uv|$, and
\item  $(3)$ is the value accumulated by all active slave automata on the word $v$, i.e.,
between the positions $|u|$ and $|uv|$ in $uvw$.
\end{itemize}

\noindent Note that $(1)'$ is bounded by $(1)$,
$(3)$ is bounded by  $2 \cdot \boundOnMulti \cdot |v|  \leq 2 \cdot \boundOnMulti^2$, and
 $(2)'  = \sum_{q \in A} \valF(q) \cdot \mult_{uv}(q)$.
We claim that $(2) - (2)' > (3)$, which means that the partial sum at
the position $|uv|$ in $uvw$ is smaller than the partial sum at the position $|u|$ in $uw$.
Indeed,  $\sum_{q \in A} \mult_{u}(q) - \mult_{uv}(q') > 2 \cdot \boundOnMulti - \boundOnMulti =  \boundOnMulti$
and for each $q \in A$, $\valF(q) \geq 4 \cdot \boundOnMulti$, therefore
$(2) - (2)'$ is at least $4 \cdot \boundOnMulti^2$, which is greater than $(3)$. 

It follows that aforementioned transformation, even applied infinitely many times,
will not increase the value of the resulting word. Therefore, 
for every run of $\nestedA$ of value $\lambda$, there is $s_0$ and an
a run of the value not exceeding $\lambda$ such that 
at each position $s > s_0$ at most $2 \cdot \boundOnMulti$ will accumulate value greater than $4\cdot \boundOnMulti$.
\end{proof}

\begin{proof}[Proof of Lemma~\ref{l:bound-on-slave}]
We transform the automaton $\nestedA$ to an equivalent deterministic 
$(\flimavg;\fsum^+)$-automaton $\nestedA_0$ 
for which runs that at each position (1)~at most  $\BoNumberOfSA$ slave automata run,
and (2)~if the master automaton takes a silent transition, each slave automaton 
takes a silent transition, are optimal.
Due to Lemma~\ref{lim-avg-technical} runs for which 
eventually at most $2 \cdot \boundOnMulti$ slave automata accumulate value greater than $4\cdot \boundOnMulti$ are optimal for $\nestedA$.
Moreover, by Lemma~\ref{many-non-silent} runs in which at least one in every  $\confNum{\nestedA}$ transitions is non-silent 
are optimal for $\nestedA$.

We define an automaton $\nestedA_0$ by modifying $\nestedA$ in two ways.
First, we extend the input alphabet to include the marking of the position $s_0$ past which
at most $2 \cdot \boundOnMulti$ slave automata accumulate value greater than $4\cdot \boundOnMulti$ are optimal for $\nestedA_0$.
Prior to that marking, a modified automaton starts only dummy slave automata that immediately terminate.
Past that marking $\nestedA_0$ simulates $\nestedA$.
Second, we modify each slave automaton of $\nestedA$ in such a way that is runs as long as it can accumulate
the value exceeding $4\cdot \boundOnMulti$. More precisely,
the master automaton starts only automata that return values exceeding
$4 \cdot \boundOnMulti$. For other slave automata it ``guesses''
their value from the set $\{0, \ldots, 4\cdot \boundOnMulti\}$ and runs a dummy automaton that takes only a single transition of this weight.
As it is deterministic, we assume that the ``guess'' is encoded in the input word.
Started slave automata run as long as they can accumulate the value 
exceeding $4 \cdot \boundOnMulti$. Once a slave automaton guesses that this is not possible,
it takes a transition of the weight $4 \cdot \boundOnMulti$ and terminates.
Again, that ``guess'' is encoded in the input word, therefore the master automaton is
able to verify that this ``guess'' is correct.

The automaton $\nestedA_0$ simulates the runs 
of $\nestedA$ past the position $s_0$ and each running slave automaton accumulates the value exceeding $4\cdot \boundOnMulti$.
Therefore, there  is a run of $\nestedA_0$ on a word corresponding to $w_{\mathrm{opt}}$
such that at most $2 \cdot \boundOnMulti +1$ slave automata run concurrently.
The automaton $\nestedA_0$ is equivalent to $\nestedA$, as the return values of slave automata
past the position $s_0$ are the same and the $\flimavg$ value function does not depend
on finite prefixes. 
Therefore, runs satisfying conditions (1) and (2) are optimal for $\nestedA_0$.
\end{proof}

The size of the automaton $\nestedA_0$ from Lemma~\ref{l:bound-on-slave} is polynomial
in the size of the master automaton of $\nestedA$.

\smallskip\noindent{\bf Proof of Step~2.} We now prove Step~2 which basically involves
the classic power-set construction.  

\smallskip\noindent{\em Construction of $\nonnestedA$.}
We adapt the classic power-set construction to 
construct a $\silent{\flimavg}$-automaton $\nonnestedA$ that simulates
runs of a given $(\flimavg; \fsum^+)$-automaton $\nestedA$ with bounded multiplicities 
and synchronized silent transitions.
Let $\nestedA$ be a deterministic $(\flimavg; \fsum^+)$-automaton. 
Without loss of generality (Lemma~\ref{l:bound-on-slave}), we assume that 
runs with (1)~multiplicities bounded by $\BoNumberOfSA$ and (2)~synchronized silent transitions are optimal for $\nestedA$.
The automaton $\nonnestedA$, which simulates runs satisfying (1) and (2), 
keeps track of the current configuration and the multiplicity of $\nestedA$. 
That is, the set of states of $\nonnestedA$ is $Q_m \times \BoNumberOfSA^{\Qslv}$, where $Q_m, \Qslv$ are respectively 
the set of states of the master automaton and the union of the sets of states of 
all slave automata of $\nestedA$. The component $\BoNumberOfSA^{\Qslv}$ encodes a part of configuration and multiplicity.
For all states  $(q_1, h_1), (q_2, h_2)$ of $\nonnestedA$ and every letter $a$,
$\nonnestedA$ has a transition $\triple{(q_1, h_1)}{a}{(q_2, h_2)}$
iff the master automaton has a transition $(q_1, a, q_2)$ with the label $i$
and the multiplicities $h_2$ follow from $h_1$ according to transitions of the slave automata
and invocation of $\slaveA_i$.
The weights of transitions of $\nonnestedA$ are defined as follows. 
If the transition $(q_1, a, q_2)$ of the master automaton of $\nestedA$ is silent,
the transition $\triple{(q_1, h_1)}{a}{(q_2, h_2)}$ is silent.
Otherwise, the weight of $\triple{(q_1, h_1)}{a}{(q_2, h_2)}$ 
equals the sum of weights of the corresponding transitions of simulated slave automata
multiplied by their multiplicities $h_1$.
Recall that the simulated runs have silent transitions synchronized,
therefore, the values accumulated by each slave automaton and the corresponding simulated automaton
are equal.

Given a run of $\nonnestedA$ one can construct a run of $\nestedA$ with multiplicities 
bounded by $\BoNumberOfSA$, and vice versa. 
Hence, for a run of $\nonnestedA$ (resp. $\nestedA$) we refer 
to the corresponding run of $\nestedA$ (resp. $\nonnestedA$).
We can solve the emptiness problem for $\nonnestedA$ as 
$\silent{\flimavg}$-automata behave similarly to $\flimavg$-automata:

\begin{restatable}{lemma}{SilLimAvg}
(1)~The emptiness problem for $\silent{\flimavg}$-automata is in $\NLOGSPACE$.
(2)~For every $\silent{\flimavg}$-automaton $\nonnestedA$ that recognizes a non-empty language,
there is a run $\eta$ of $\nonnestedA$ such that 
(a)~the value of $\eta$ is minimal among values of all runs of $\nonnestedA$,
(b)~at least one in every $|\nonnestedA|$ transitions is non-silent, and 
(c)~partial sums converge, i.e.,
\[\lim\sup_{k \rightarrow \infty} \frac{1}{k} \cdot \sum_{i=1}^{k} (\cost(\eta))[i] =
\lim\inf_{k \rightarrow \infty} \frac{1}{k} \cdot \sum_{i=1}^{k} (\cost(\eta))[i] \]
\label{silent-lim-avg}
\end{restatable}
\begin{proof}
Let $\aut = (\Sigma, Q,q_0, \delta, F,\cost)$ be a $\silent{\flimavg}$-automaton. 
We show that if there is a run of $\aut$ of the value not exceeding $\lambda$
there is a lasso run such that the average weight in its cycle does not exceed $\lambda$, i.e.
there is a run $\pi = \pi[0] \pi[1] \ldots \pi[n]$ of length $n \leq |\lambda|$ on a word $w = w[1]\ldots w[n]$ 
such that $\pi[0] = q_0$, for some $i < n$ we have $\pi[i] = \pi[n]$ (such a run $\pi$ is called a lasso),
$\pi[0], \ldots, \pi[n-1]$ are distinct and $\frac{1}{n-i} (\cost(\pi[i], w[i+1], \pi[{i+1}]) + \ldots + 
\cost(\pi[{n-1}], w[{n}], \pi[n])) \leq \lambda$.
The existence of such a lasso can be decided in $\NLOGSPACE$, which shows (1).
As there are only finite number of values of lassos, there is $\lambda_0$ which is the smallest. 
It follows that there is no run of $\aut$ of the value $\lambda'$ smaller than $\lambda_0$;
otherwise there would be a lasso of the value not exceeding $\lambda'$. 
Thus, the lasso of the value $\lambda_0$ has the minimal value among values of all runs, and
the sequence of partial averages converge. 

First, we define a $\flimavg$-automaton $\aut_{\mathrm{fix}}$ such that 
for every accepting run $\eta$ of $\aut$, the run $\eta'$, resulting from $\eta$ by
removing silent transitions, is an accepting run of $\aut_{\mathrm{fix}}$, and vice
versa, every accepting run of $\aut_{\mathrm{fix}}$ can be extended to a run of $\aut$
by inserting silent transitions.
The set of states of $\aut_{\mathrm{fix}}$ is the same as $\aut$ and
the transition relation of $\aut_{\mathrm{fix}}$ consists of 
$(q_1, a, q_2)$ such that 
there is $(q_1', a, q_2') \in \delta$  and
$q_1'$ (resp. $q_2$) is reachable from $q_1$ (resp. $q_2'$) by a path consisting of only silent transitions.
The weight of such a transition is the infimum over the weights of transitions $(q_1', a, q_2') \in \delta$
that generate $(q_1, a, q_2)$.
It follows from the construction that $\aut_{\mathrm{fix}}$ has the stipulated properties and
their corresponding runs have the same value. 
Therefore, $\aut_{\mathrm{fix}}$ has a lasso $l_0$ such that the average weight in its cycle does not exceed $\lambda$.
The lasso $l_0$ can be extended to a lasso $l_1$ in $\aut$, which can have states that occur 
multiple times. We can remove duplicate states in the following way, if the average weight between $i$ and $j$
with $l_1[i] = l_1[j]$ is above $\lambda$, we remove all the states between $i+1$ and $j$.
Otherwise, if the average weight does not exceed $\lambda$, we can remove all the states following $\pi[j]$.
Hence, we have shown that if $\aut$ has a run of the value $\lambda$ it has 
a lasso such that the average weight in its cycle does not exceed $\lambda$.
\end{proof}

\smallskip\noindent{\bf Proof of Step~3.} 
We now prove Step~3, i.e., we show that the emptiness problems for $\nestedA$
and $\nonnestedA$ coincide. The main problem is that, even though a run of $\nestedA$ and the corresponding 
run of $\nonnestedA$ represent the same sequences of weights, they can have different values. 
Still, we show that infima over values of all runs of $\nestedA$ and $\nonnestedA$
are equal.

\smallskip\noindent{\em Accumulation of weights.}
The automata $\nestedA$ and $\nonnestedA$ compute their values differently. 
In $\nestedA$ a slave automaton started at a position $k$ computes its value
and returns it as a weight of the transition at position $k$, whereas 
in $\nonnestedA$,  simulated slave automata run concurrently and add their weights to
the partial sum at each step.
To visualize this, consider a matrix such that the value at $(i,j)$ is
the weight of the transition of $j$-th slave automaton at position $i$ in the input word.
Then, the value of $\nestedA$ is the limit average of the sums of rows,
whereas the value of $\nonnestedA$ is the limit average of the sums 
of columns.
\medskip

\begin{figure}[h]
\begin{center}
\noindent\begin{tikzpicture}
\tikzAccummulation
\end{tikzpicture}
\caption{The matrix depicting the difference in aggregation of weights by 
$\nestedA$ and $\nonnestedA$.}
\label{f:aggregation}
\end{center}
\end{figure}


In consequence, a run of $\nestedA$ and
the corresponding run of $\nonnestedA$ can have different values.
\begin{example}
Consider the automaton $\nestedAMod$ that has been discussed above as a modification of
an automaton from Example~\ref{e:OptimalUnbounded2} and a word 
$w = bab\ldots ba^{2^k-1}b\ldots$. 
At position $2^k+1$, the partial sum $\sum_{i=1}^{k} (\cost(\pi_i))$
is equal to the sum of (1)~the values of blocks $bab, \ldots, ba^{2^{k-1}-1}b$
and (2)~the value of $\slaveA_b$ started at the beginning of $ba^{2^{k}-1}b$ equal $2^{k}-1$. 
The value of a block $ba^{2^{i}-1}b$ equals $2^{i}-1$ (accumulated by $\slaveA_b$) plus $2^{i-1}-1$
(accumulated by $\slaveA_{a1}$). 
Thus, the value (1)~equals $2^{k}+2^{k-1}-2\cdot k$, i.e., 
 $2^{k}+2^{k-1} - O(k)$.
Therefore, the partial average at position $2^k+1$ is $\frac{5}{2}-O(\frac{k}{2^k})$.
Thus, the value of $w$ assigned by $\nestedAMod$ is $\frac{5}{2}$.
In contrast, the automaton $\nonnestedAMod$ simulating $\nestedAMod$ simply
takes a transition of weight $0$ on letter $b$, $1$ on even occurrences of letter $a$
and $2$ on odd occurrences of letter $a$. 
The value of $w$ assigned by $\nonnestedAMod$ is $\frac{3}{2}$.
\end{example}

Nevertheless, we show that the emptiness problems for $\nonnestedA$ and 
$\nestedA$ coincide. In the following lemma we will use a notion of \emph{reset words}
used to terminate long runs of slave automata. 

\smallskip\noindent{\em Reset words.} 
Given a word $w$, a finite word $u$ is a \emph{reset word} for 
$\nestedA$ at position $i$ if in $w[1,i] u w[i+1, \infty]$ we have (1)~the configuration of $\nestedA$
at positions $i$ and $i + |u|$ is the same (recall that $\nestedA$ is deterministic), 
and  (2)~all slave automata active at position $i$ terminate
before position $ i + |u|$. 
We say that the word $w[1,i] u w[i+1, \infty]$ is the result of injecting
a reset word at position $i$.
Observe that (1) implies that $\nestedA$ accepts $w$ iff it accepts $w[1,i] u w[i+1, \infty]$.
As in an accepting run of $\nestedA$ past some finite position all configurations 
occur infinitely often and all slave automata terminate after finite number of steps, 
therefore at almost all positions $p$, there exists a reset word that can be injected at $p$.
Basically, that word occurs already at position $p$.
In addition, by simple (un)pumping argument we can show that for almost all positions 
there exist reset words with length bounded by $|\Qslv| \cdot \confNum{\nestedA}$.

\begin{lemma}
The emptiness problems for $\nestedA$ and $\nonnestedA$ coincide.
\label{l:emptinessCoincide}
\end{lemma}
\begin{proof}
Observe that for every run $(\Pi, \pi_1, \pi_2, \ldots)$ of $\nestedA$
with at most $\BoNumberOfSA$ concurrently running slave automata,
and the corresponding simulation run $\eta$ of $\nonnestedA$ at every position $k$
we have
$\sum_{i=1}^{k} (\cost(\eta))[i] \leq \sum_{i=1}^{k} (\cost(\pi_i))$.
Therefore, the infimum over runs of $\nonnestedA$ 
does not exceed the infimum over runs of $\nestedA$.

Conversely, due to Lemma~\ref{silent-lim-avg}, 
$\nonnestedA$ has a run $\eta$ of the minimal value and whose sequence of partial averages converges.
First, we show that there is a run $\eta'$ of the same value as $\eta$ such that 
in the corresponding run of $\nestedA$, for every $k$ greater than some constant, 
each slave automaton started before position $k$, terminates before position $k + \log(k)$.
Then, we show that $\eta'$ and the run of $\nestedA$ corresponding to $\eta'$ have the same value. 

Consider a sequence $\{a_i\}_{i \geq 0}$ where $a_0$ is the least position 
past which every configuration occurs infinitely often and $a_{i+1} = a_i + \log a_i$. 
Observe that 
$|\{ a_i : a_i < k \}| = o(k)$, i.e., $\lim_{k\rightarrow}\infty \frac{|\{ a_i : a_i < k \}|}{k} = 0$.
Let $\eta'$ be a run obtained from $\eta$ by injecting reset words on 
positions $a_i$ in $\eta$.
The value of $\eta'$ is the same as the value of $\eta$. 
 Indeed,  for almost all $k$ we have
\[  \sum_{i=1}^{k} (\cost(\eta'))[i]  \leq 
\sum_{i=1}^{k} (\cost(\eta))[i] + o(k)
\]
It follows from the fact that 
there are $|\{ a_i : a_i < k \}| = o(k)$ reset words up to position $k$
and the total increase of the partial sum due to each of them is bounded by the product of (1),(2),(3), where
(1) is the maximal length of a reset word, (2) is the number of currently running slave automata,
and (3) is the maximal weight a slave automaton can take. Note that (1),(2),(3)
are bounded by a constant, hence up to position $k$, the total increase of the partial sum due to injected reset words 
is $o(k)$. 
Due to Lemma~\ref{silent-lim-avg}, there are at least $\frac{k}{|\nonnestedA|}$
non-silent transitions up to position $k$. Hence the values of $\eta'$ do not exceed the 
value of $\eta$, which is minimal. Hence, the values of $\eta, \eta'$ are equal. 

Now, we consider a run of $(\Pi, \pi_1, \pi_2, \ldots)$ of $\nestedA$ 
that corresponds to $\eta'$. Observe that each slave
automaton started at position $k$, terminates after at most $\log k$ steps.
Therefore, the partial sum of weights of $(\Pi, \pi_1, \pi_2, \ldots)$
up to $k$ is bounded by the partial sum of weights in $\eta'$
up to $k + \log k$, i.e.,
for almost all $k$ we have
\[ \sum_{i=1}^{k} (\cost(\eta'))[i] \leq \sum_{i=1}^{k} (\cost(\pi_i)) \leq \sum_{i=1}^{k + \log k} (\cost(\eta'))[i] \]
However, each $(\cost(\eta'))[i]$ is bounded by a constant.
Therefore, $\sum_{i=k+1}^{k + \log k} (\cost(\eta'))[i] = O(\log k)$.
Again, as the number of non-silent transitions up to position $k$ is $\Omega(k)$,
we have $\lim_{k \rightarrow \infty} \frac{O(\log k)}{\Omega(k)} = 0$.
Since $\nestedA$ and $\nonnestedA$ take non-silent transitions at the same positions, 
the limit averages of $(\cost(\eta'))[1], (\cost(\eta'))[2], \dots$ 
and $\cost(\pi_1), \cost(\pi_2), \dots$ are equal. 
Thus, the value of the infimum over runs of $\nonnestedA$ 
does not exceed the value of the infimum over runs of $\nestedA$.
This implies that the emptiness problems for $\nonnestedA$ and $\nestedA$
coincide. 
\end{proof}

Finally, we prove the main statement.

\begin{lemma}
The emptiness problem for  $(\flimavg; \fsum^+)$-automata is in $\EXPSPACE$.
\end{lemma}
\begin{proof}
First, we transform a given $(\flimavg; \fsum^+)$-automaton $\nestedA$ 
to a deterministic $(\flimavg; \fsum^+)$-automaton $\nestedA^d$ (Lemma~\ref{WLOGDeterministicLimAvgSum}). 
The transformation does not change the cardinality of the set of configurations, i.e.,
$\conf(\nestedA) = \conf(\nestedA^d)$.
Next, we transform $\nestedA^d$ to an equivalent automaton $\nestedA'$ for
which runs with (1)~multiplicities bounded by $\BoNumberOfSA$ 
and (2)~synchronized silent transitions are optimal (Lemma~\ref{l:bound-on-slave}).
Finally, we define a $\silent{\flimavg}$-automaton $\nonnestedA$ that simulates 
$\nestedA'$ on runs that satisfy (1) and (2).
Since the emptiness problems for 
$\nestedA'$ and $\nonnestedA$ coincide (Lemma~\ref{l:emptinessCoincide}), we solve the former.

Observe that the size of $\nonnestedA$ is doubly exponential in the size of $\nestedA$.
Indeed, let $\Qslv'$ be the disjoint union of the sets of states of the slave automata of $\nestedA'$.
The size of $\Qslv'$ is exponential in the size of $\nestedA$, therefore $|\BoNumberOfSA^{\Qslv'}|$
is doubly-exponential.
However, the emptiness problem for $\nonnestedA$ is decidable in $\NLOGSPACE$ (Lemma~\ref{silent-lim-avg}).
Hence, the emptiness problem for $(\flimavg; \fsum^+)$-automata is in $\EXPSPACE$.
\end{proof}

\begin{remark}
\label{rem:secondParRemark}
The transformations from Lemma~\ref{WLOGDeterministicLimAvgSum} and Lemma~\ref{l:bound-on-slave}
are polynomial in the size of the master automaton of a given nested automaton $\nestedA$. Then, the number of configurations
and, in consequence, $\BoNumberOfSA$ are polynomial in the size of the master automaton of $\nestedA$. Finally,
the set of states of the $\silent{\flimavg}$-automaton $\nonnestedA$ constructed in Step~2 is
$Q_m \times \BoNumberOfSA^{\Qslv}$. Thus, it is polynomial in the size of the master automaton of $\nestedA$ as well as 
its weights. Therefore, the emptiness problem of $(\flimavg; \fsum^+)$-automata is in $\PTIME$ if the total size of slave automata is bounded by a constant.
\end{remark}

\section{Related work}\label{sec:related}
In this section we discuss various related work.

\smallskip\noindent{\em Weighted counterpart of Boolean nested automata.}
Weighted nested automata have been considered 
in~\cite{DBLP:conf/icalp/BolligGMZ10} in the context of finite words, 
where the weights are given over semirings.
 It is further required
that the semirings of all master and slave automata coincide, while in
our case, the value functions may differ.
Since the semirings of master and slave automata must coincide for~\cite{DBLP:conf/icalp/BolligGMZ10}, 
it can be interpreted as defining weighted counterpart of Boolean nested automata over finite words.
Adding nesting structure to words and trees have been extensively studied for 
non-weighted automata in~\cite{ACM06,AM09} and also applied to 
software model checking~\cite{ACM11}.
The work of~\cite{DBLP:conf/icalp/BolligGMZ10} defines a weighted
counterpart of nesting of finite words, whereas we consider nesting of various weighted 
automata to express properties of infinite behaviors of systems.
Properties such as long-run average response time cannot be expressed in the framework of~\cite{DBLP:conf/icalp/BolligGMZ10}.

\smallskip\noindent{\em Weighted MSO Logics.} 
Quantitative properties can be expressed in weighted logics, for example, 
\emph{Weighted MSO Logics}~\cite{DBLP:journals/iandc/DrosteM12} and 
weighted temporal logic~\cite{DBLP:conf/lics/BokerCHK11}.  For the
basic decision problems for weighted logics over infinite words, the
reduction is to weighted automata.  For a given set of value functions
that assigns values to infinite runs (such as $\FinVal$ and
$\InfVal$), weighted MSO logics are as expressive as weighted automata
with the same class of value
functions~\cite{DBLP:journals/iandc/DrosteM12}.  It follows that with
$\FinVal$ and $\InfVal$ as the value functions, weighted MSO logic
cannot express the average response property.  One can express average
response time in weighted MSO logics by adding average response time itself
as a primitive value function in the $\omega$-valuation monoid.  The
decidability of weighted MSO logics with such a primitive can be
established by a reduction to weighted automata that are able to
express average response time, such as nested weighted automata.
However, the reduction is non-elementary, as the basic decision
problems for even non-weighted MSO logic have non-elementary
complexity, whereas our complexity results range from
$\PSPACE$-complete to $\EXPSPACE$.

\smallskip\noindent{\em Register automata.}
Another related model for specifying quantitative properties
are \emph{register automata}~\cite{DBLP:conf/lics/AlurDDRY13}, which
are parametrized by cost functions.  The main differences
between~\cite{DBLP:conf/lics/AlurDDRY13} and nested weighted automata
are as follows: (i)~register automata are over finite words, whereas
we consider infinite words, and (ii)~we consider nested control of
automata, whereas register automaton are non-nested.  As a result,
both in terms of expressiveness and decidability results nested weighted
automata are very different from register automata.  For example, the
emptiness of register automata with max and sum value functions is
decidable, while we show emptiness to be undecidable for deterministic
nested weighted automata with these value functions.

\smallskip\noindent{\em Other related models.} 
Other possible quantitative models are visibly pushdown automata (VPA) 
with limit-average functions, or quantitative models consider in~\cite{ICALP08}.
The framework of~\cite{ICALP08} neither captures the average response time 
property nor presents any decidability results. 
For VPA with limit-average functions it follows from~\cite{CV12} that even 
perfect-information VPA games (that correspond to simulation) with limit-average
objectives are undecidable (the undecidability proof of~\cite{CV12} is for general
pushdown games, but the proof itself also works for VPA games).
Thus though there exist many other quantitative models, there exists no 
framework that can express the average response time property and have elementary-time 
complexity algorithms for the basic decision problems.

\section{Conclusion and Future Work}

Motivated by important system properties such as long-run average
response time, we introduced the framework of nested weighted automata
as a new, expressive, and convenient formalism for specifying
quantitative properties of systems.  We answered the basic decision
questions for nested weighted automata.  There are several directions
for future work.  First, we have an open conjecture
(Conjecture~\ref{conj1}) regarding the decidability of the emptiness
of $(\flimavg;\fsum)$-automata.  Second, another
interesting direction would be to establish optimal complexity results
for the decision problems.  Third, there are several possible
extensions of the nested weighted automaton model, such as (i)~two-way
master and slave automata, (ii)~multiple levels of nesting, and
(iii)~instead of infimum across paths consider average measures across
paths (i.e., probability distributions over runs and expected value of
the runs, as in~\cite{DBLP:conf/concur/ChatterjeeDH09}).

\bibliography{papers-ml}

\end{document}